\documentclass[a4paper,superscriptaddress,twocolumn,pra,showkeys,nofootinbib]{revtex4-1}

\usepackage{graphicx} 
\usepackage{epstopdf}
\usepackage{color} 
\usepackage{amsmath} 		
\usepackage{amssymb}  
\usepackage{amsthm}
\usepackage{capt-of} 
\usepackage{natbib}
\usepackage[percent]{overpic}

\newtheorem{theorem}{Theorem}[section]

\newtheorem{proposition}[theorem]{Proposition}
\newtheorem{corollary}[theorem]{Corollary}

\newenvironment{definition}[1][Definition]{\begin{trivlist}
\item[\hskip \labelsep {\bfseries #1}]}{\end{trivlist}}

\newenvironment{remark}[1][Remark]{\begin{trivlist}
\item[\hskip \labelsep {\bfseries #1}]}{\end{trivlist}}

\begin{document}

\title{Flag-based Control of Orbit Dynamics in Quantum Lindblad Systems}

\author{Patrick Rooney}
\email{darraghrooney@gmail.com}
\noaffiliation

\author{Anthony M. Bloch}
\email{abloch@umich.edu}
\affiliation{Department of Mathematics, University of
Michigan, Ann Arbor, MI 48109}

\author{C. Rangan}
\email{rangan@uwindsor.ca}
\affiliation{Department of Physics, University of Windsor, ON,
N9B 3P4, Canada}

\begin{abstract}
In this paper, we demonstrate that the dynamics of an $n$-dimensional Lindblad control system can be separated into its inter- and intra-orbit dynamics when there is fast controllability. This can be viewed as a control system on the simplex of density operator spectra, where the flag representing the eigenspaces is viewed as a control variable. The local controllability properties of this control system can be analyzed when the control-set of flags is limited to a finite subset. In particular, there is a natural finite subset of $n!$ flags that are effective for low-purity orbits.
\end{abstract}


\maketitle

\section{Introduction}
Advances in quantum technologies, such as the nascent  progress in quantum computation \cite{Feynman82}\cite{NielsenChuangBook}\cite{RanganBucksbaum01}\cite{PalaoKosloff02}, as well as the developments of coherent control of chemical reactions  \cite{ShapiroBrumer86}\cite{TannorRice85} and NMR \cite{ErnstetalBook}, have resulted in great effort to apply mathematical control theory \cite{SontagBook} to quantum mechanical systems \cite{DAlessandroBook}. The interaction of a system with its environment is a major obstacle in quantum control, and as a result quantum control theory has expanded from closed systems \cite{HuangTarnClark83} to open systems (see \cite{MabuchiKhaneja2005}, \cite{BrifChakrabartiRabitz2010}, \cite{DongPetersen2011} and \cite{AltafiniTicozzi2012} for surveys). 

A common method of modeling open systems is to assume they are Markovian and time-independent, in which cases the dynamics are described by a quantum dynamical semi-group and the Lindblad master equation \cite{Lindblad76}\cite{GoriniKossakowskiSudarshan76}\cite{BreuerPetruccioneBook}. Typically, the control functions appear in the system Hamiltonian (although there has been progress in engineering Lindblad dynamics \cite{LloydViola01}\cite{Baconetal01}\cite{Barreiroetal2011}). This means that, absent the interaction with the environment, the controls are only capable of steering the system within a given unitary orbit \cite{TannorBartana99}\cite{SklarzTannorKhaneja04}\cite{Schirmeretal04}. The motion between orbits depends on the Lindblad super-operator. Consequently, the Hamiltonian cannot directly affect the eigenvalues, or the purity $Tr(\rho^2)$, since the eigenvalues of the density operator are constant on any orbit. If
the optimal time \cite{KhanejaGlaserBrockett02} between two unitarily equivalent density operators is much smaller than the time-scale characterized by the Lindblad dynamics, it becomes an interesting question as to how best position the system on any given orbit.

The aim of this paper is to formally consider an approach to control of open quantum systems in which the space of density matrices is decomposed into spectra (the set of possible orbits) and flags (the positions along a given orbit). If one has sufficiently fast and complete Hamiltonian control, the intra-orbit dynamics can be made arbitrarily faster than the inter-orbit dynamics, since the Lindblad super-operator is bounded. After separating the dynamics, we want to view the flag trajectory as a control function, and the spectrum as the state variables. We refer to this viewpoint as flag-based control. After a desired flag trajectory has been determined, we can consequently reconstruct the necessary Hamiltonian, which contains the true control functions. We are building on previous work on two-dimensional systems \cite{us_nis2_a}\cite{us_nis2_b}. The $n=2$ case is easier to study from a control perspective as the set of orbits is isomorphic to a closed line segment, and all orbits but one are isomorphic to a sphere. In order to generalize to $2<n<\infty$, one must address the delicacies of dealing with more complicated orbit sets, as well as cope with the difficulties that come with non-trivial control sets. Chapter 8 in reference \cite{BengtssonZyczkowskiBook} discusses the geometry of density matrices, and, in particular, their orbit sets. Our approach contrasts with the generalized Bloch vector representation approach \cite{Schirmeretal04}\cite{SchirmerWang2010}, which yields an affine differential equation on the vector space of density operators. This representation has little to do with the orbit structure however. 

One obstacle that arises in our approach is the non-linearity of the flag-set. The flag-set is always the quotient manifold $U(n)/(U(m_1)\times\cdots\times U(m_\alpha))$, where $m_\alpha$ is the multiplicity of the $\alpha$th eigenvalue of the density operator. It is therefore non-trivial to apply standard control theory results to a flag-based control system. In this paper, we demonstrate that a local controllability result can be applied when one limits the flag-controls to a finite subset of the flag-set. In particular, the behavior of the Lindblad operators at the completely mixed state yields a natural set of $n!$ flags that are particularly useful for low-purity orbits.

Infinite-dimensional quantum systems \cite{BlochBrockettRangan10} present many technical difficulties. In particular, the Lindblad super-operator is not necessarily bounded, which means it has no characteristic time-scale, and we cannot assume our unitary control is faster than the Lindblad dynamics. For this reason, we consider only finite-dimensional systems. Additionally, while there is considerable research in using feedback to control both closed \cite{MirrahimiRouchon09} and open \cite{JamesGough10}\cite{BoutenVanHandelJames09} quantum systems, we shall only consider the open-loop case, where there is no feedback.

In section II, we decompose the Lindblad master equation into its spectral and flag components, and in section III, we re-interpret the spectral ODE as a control equation. In section IV, we analyze the local controllability of finite flag control-sets, and in section V we show some examples.

\section{Separation of Spectral and Flag Dynamics}

A state in an $n$-dimensional  open quantum system is described by an operator $\rho$ on the $n$-dimensional Hilbert space, called the density operator. It must be positive semi-definite with unit trace. It can be written in terms of its eigenvalues: 
\begin{equation}
\rho = \sum_{\alpha=1}^{n_d} \lambda^d_\alpha P_\alpha, \label{decomp}
\end{equation}
where $\Lambda^d = \{\lambda^d_\alpha: \alpha = 1, \cdots, n_d\le n\}$ is the set of \emph{distinct} eigenvalues of $\rho$, and $\{P_\alpha\}$ are orthogonal projectors onto the corresponding eigenspaces\footnote{Henceforth, Greek indices will be used for summing over distinct eigenvalues, while Latin indices will be used for eigenvalues with multiplicity. Moreover, all projectors are assumed to be orthogonal.}. The properties of $\rho$ demand that all eigenvalues lie on the interval $[0,1]$ and $\sum_\alpha m_\alpha \lambda^d_\alpha= \sum_j\lambda_j = 1$, where $m_\alpha$ is the multiplicity of $\lambda^d_\alpha$.

The dynamics of a system with Lindblad dissipation is described by the Hamiltonian $H(t)$, which is a (possibly time-dependent) Hermitian operator, and a set of $N$ Lindblad operators $\{ L_k \}$ with the Lindblad equation:
\begin{align}
\frac{d}{dt}\rho(t) &= \mathcal{L}(\rho(t)) := [-iH(t), \rho(t)] + \mathcal{L}_D(\rho(t)) \label{lindblad}\\
\mathcal{L}_D(\rho) &:= \sum_{k=1}^N \left( L_k \rho L_k^\dagger - \frac{1}{2} \{L_k^\dagger L_k, \rho \} \right),
\end{align}
where the braces in the Lindblad superoperator $\mathcal{L}_D$ indicate an anti-commutator. 

We are interested in investigating and controlling how a system moves between unitary orbits. In the absence of Lindblad dissipation, the solution to (\ref{lindblad}) can be written $\rho(t) = U(t)\rho(0) U(t)^\dagger$ where $U(t)$ is a trajectory on the unitary group $U(n)$ obeying $\frac{d}{dt}U(t) = -iH(t)U(t)$. Since $U(t)$ is unitary, the eigenvalues of $\rho(t)$ are invariant under the Hamiltonian evolution. That is, if we define the unitary orbit $\mathcal{O}(\rho) := \{U\rho U^\dagger: U\in U(n) \}$, the system does not leave the orbit without the influence of $\mathcal{L}_D$. For simplicity, we will assume fast controllability on the orbit: we can write 
\begin{equation}
H(t) = H_0+ \sum_{i=1}^{n^2-1} u_i(t) H_i \label{controlham},
\end{equation}
where $\{H_i: i=1,2,\dots,n^2-1\}$ is a basis of $\mathfrak{su}(n)$, and the $\{u_j(t)\}$ are real-valued control functions that are unbounded and piecewise-continuous. The unboundedness is a key property: since $\mathcal{L}_D()$ is a bounded super-operator, motion along a unitary orbit can be made arbitrarily faster than motion between orbits. And because $\{H_i\}$ span the Lie algebra, any point on the orbit is reachable from any other.

We want to separate the dynamics of the eigenvalues from that of the projectors. We must make a distinction between three objects: the unordered set of $n_d$ distinct eigenvalues $\Lambda^d$, a vector $\Lambda$ of possibly repeated eigenvalues, and the unordered multiset\footnote{A \emph{multiset} is a set that may contain repeated elements} $\Lambda_Q$ of $n$ possibly repeated eigenvalues. The expression (\ref{decomp}) requires $\Lambda^d$. The space of unitary orbits however is in one-to-one correspondence with the space of multisets $\Lambda_Q$. Moreover, we want to write down a linear ODE for the eigenvalues, which necessitates using the vector $\Lambda$. 

We use the subscript $Q$ because $\Lambda_Q$ lives on a quotient space. $\Lambda$ exists on an $n$-simplex $\mathcal{T}\subset\mathbb{R}^n$, which has vertices $(1,0,\cdots,0)$, $(0,1,0,\cdots,0)$, $\cdots$, $(0,0,\cdots,0,1)$. $\Lambda_Q$ on the other hand lives on $\mathcal{T}_Q:=\mathcal{T}/S_n$, where $S_n$ is the symmetric group. Technically, $\mathcal{T}_Q$ is an orbifold with boundary: an orbifold is the quotient of a manifold with a finite group, in this case the group of eigenvalue-reorderings. Note that $\mathcal{T}_Q$ can be identified with another simplex, namely that with vertices $(1,0,\cdots,0)$, $(\frac{1}{2},\frac{1}{2},0,\cdots,0)$, $(\frac{1}{3},\frac{1}{3},\frac{1}{3},0,\cdots,0)$, $\dots$, $(\frac{1}{n},\frac{1}{n},\cdots,\frac{1}{n})$. This choice is not actually unique: $\mathcal{T}$ is essentially the union of $n!$ different subsimplices, each with a particular ordering. In the language of \textit{e.g.} \cite{BengtssonZyczkowskiBook}, $\mathcal{T}$ is the \emph{eigenvalue simplex}, and each subsimplex is a \textit{Weyl chamber}. We have chosen the Weyl chamber in which the eigenvalues are in non-increasing order. Let us call this chamber $\mathcal{T}_I$, and the other chambers $\mathcal{T}_\sigma$ corresponding to the permutations $\sigma \in S_n$.

These simplices are $(n-1)$-dimensional subsets of $\mathbb{R}^n$. It can be useful to project them onto $\mathbb{R}^{n-1}$. We consider a map $\mathcal{P}$:
\begin{align}
\bar{\mathcal{T}} &:=\mathcal{P}(\mathcal{T}) \subset\mathbb{R}^{n-1}\\
x &:= \mathcal{P}(\Lambda)\\
x_j &:= \frac{1}{\sqrt{j(j+1)}}\left(\sum_{i=1}^j \lambda_i -  j\lambda_{j+1}\right).
\end{align}
$\mathcal{P}$ is a linear map: let $\Pi$ be its corresponding $(n-1)\times n$ matrix, so that $x=\Pi\Lambda$. Let $\iota$ denote the $\Lambda$ corresponding to the completely mixed state: $\iota = \langle\frac{1}{n},\frac{1}{n},\cdots, \frac{1}{n}\rangle$. One can check the following identities: $\Pi\iota = 0$, $\Pi\Pi^T = I_{n-1}$, $\Pi^T \Pi = I_n - n \iota\iota^T$ and $\iota^T\Lambda = \frac{1}{n}\iota$. Using these identities we can see that $\Lambda = \iota + \Pi^T x$, and also that $\mathcal{P}$ is an isometry\footnote{That is, $||\Pi(\Lambda_1-\Lambda_2)|| = ||\Lambda_1-\Lambda_2||$ for any $\Lambda_1$, $\Lambda_2$.}. Therefore, $\bar{\mathcal{T}}$ is an $n$-simplex with the same side-length $\sqrt{2}$ as $\mathcal{T}$. It is also centered at the origin. Note that the $n$ faces of $\mathcal{T}$ correspond to eigenvalue zeroes $\lambda_j=0$, but this is not the case for $\mathcal{T}_I$. One of its faces corresponds to the lowest eigenvalue vanishing: $\lambda_n=0$. The remaining $n-1$ faces correspond to eigenvalue crossings $\lambda_j=\lambda_{j+1}$. 

Fig. \ref{fignew} shows $\bar{\mathcal{T}}$ for $n=3$ (top) and $n=4$ (bottom). For $n=3$, there are six Weyl chambers, and the highlighted chamber is $\bar{\mathcal{T}}_I$. The central point corresponds to the completely mixed state, the three outer vertices correspond to the orbit of pure states, and the three remaining points correspond to the orbit $\Lambda_Q=\{\frac{1}{2},\frac{1}{2},0\}$. There are three boundary edges corresponding to $\lambda_j=0$, and three inner edges corresponding to $\lambda_j=\lambda_k$, $j\ne k$.

For $n=4$, there are $24$ Weyl chambers and $\bar{\mathcal{T}}_I$ is shown in dark grey. There are six inner faces corresponding to $\lambda_j=\lambda_k$, $j\ne k$, and we have shown one in light grey. There are four outer faces corresponding to $\lambda_j=0$. In total, there are twenty-five edges of interest (many not shown). We have highlighted three: an outer edge $\lambda_3 = \lambda_4 = 0$, an inner edge $\lambda_1=\lambda_4$, $\lambda_2=\lambda_3$, and an edge inside an outer face $\lambda_2=0$, $\lambda_3=\lambda_4$. 
 
\begin{figure}
\includegraphics[width=\columnwidth]{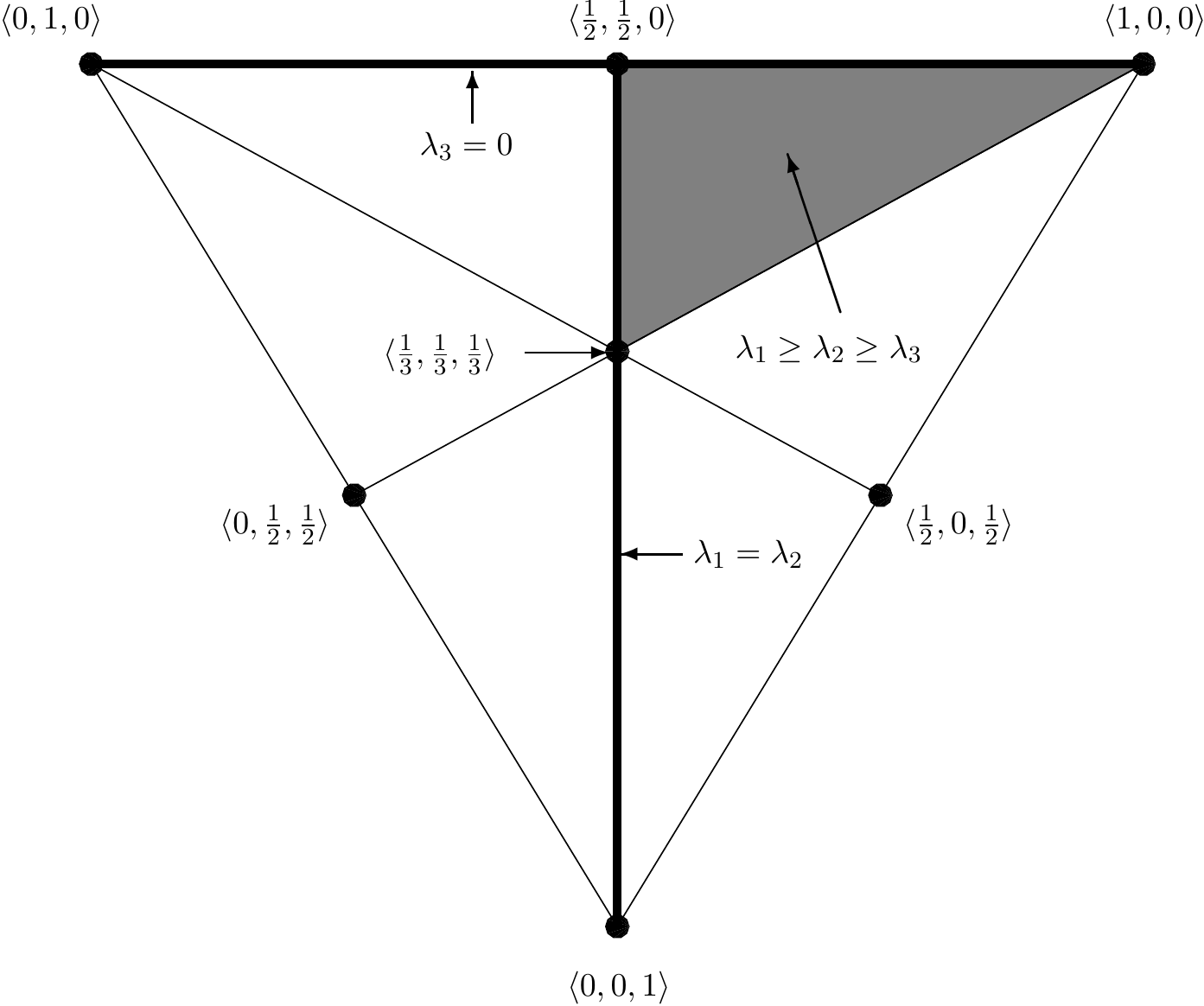}
\includegraphics[width=\columnwidth]{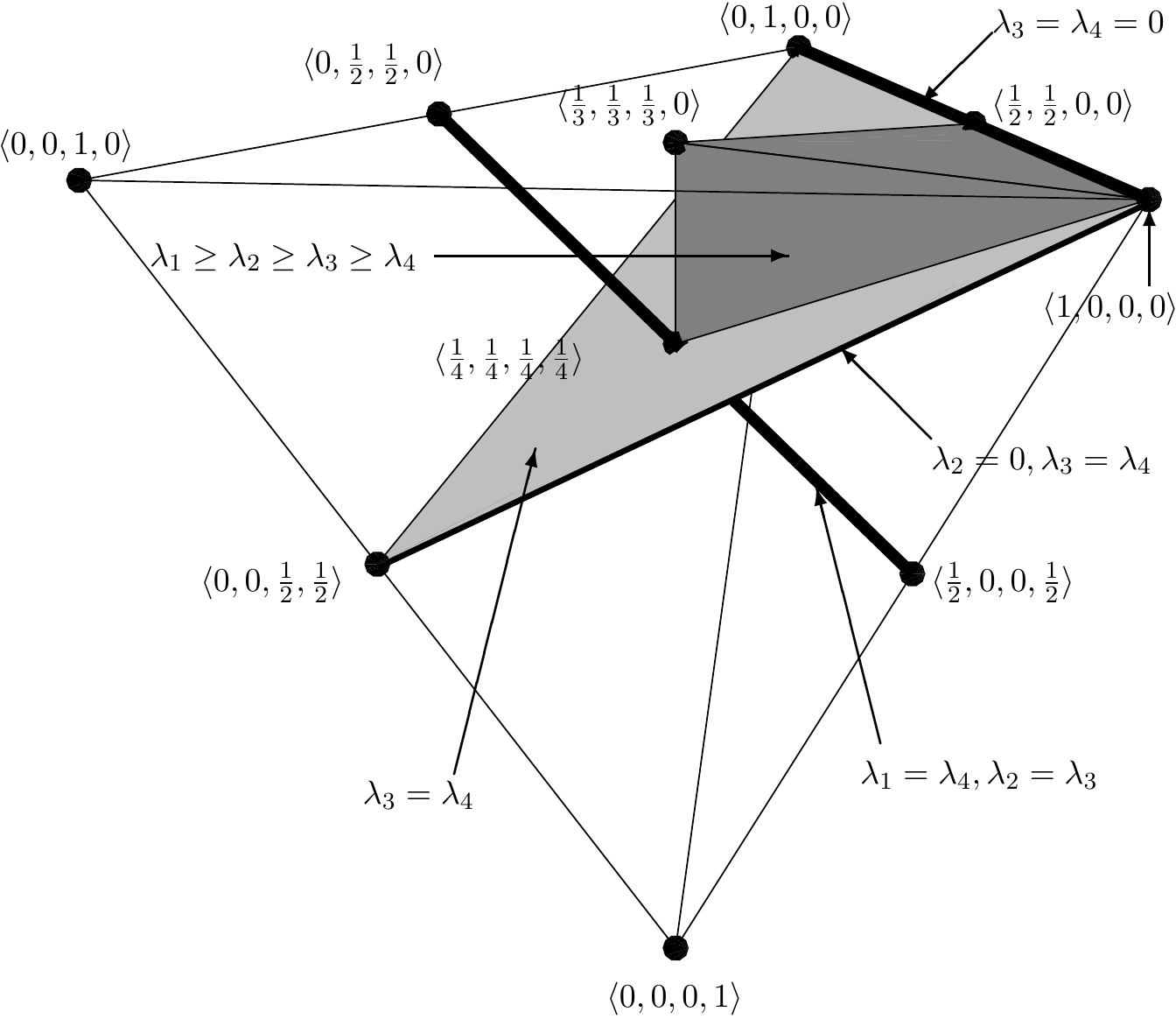}

\caption{$\bar{\mathcal{T}}$ for $n=3$ (top) and $n=4$ (bottom). We have highlighted and labeled several edges, faces and points. The Weyl chamber $\bar{\mathcal{T}}_I$ is shown in dark grey in both.  }
\label{fignew}
\end{figure}

In the remainder of this paper, we will study (differentiable) trajectories through $\mathcal{T}$ and $\mathcal{T}_Q$. We must clarify what we mean by differentiable: for any trajectory $\rho(t)$, there is one trajectory $\Lambda_Q(t)$ and several different trajectories $\Lambda(t)$ (if eigenvalues do not cross, then there are $n!$ continous $\Lambda(t)$). If we were to restrict $\Lambda(t)$ to $\mathcal{T}_I$, the trajectory would typically be non-differentiable at eigenvalue crossings. We would like to keep differentiability, and so instead of considering $\Lambda(t)\in \mathcal{T}_I$, we will consider $\Lambda(t)\in\mathcal{T}$, with the understanding that such a trajectory is not unique. When we say that $\Lambda_Q(t)$ is differentiable at time $t$, we mean that there exists a differentiable $\Lambda(t)$ that belongs to the equivalence class $\Lambda_Q(t)$. We will also refer to an eigenvalue crossing as \emph{sharp} if all crossing eigenvalues have different time derivatives at the crossing time (which implies that the crossing is isolated). Note that if differentiability holds at a sharp crossing, the relevant eigenvalues necessarily swap ordering. 
 
While $\Lambda(t)$ describes the inter-orbit motion of $\rho(t)$, the intra-orbit motion can be described by a flag. A \emph{flag} $F$ is a nesting of linear subspaces in the Hilbert space $\mathcal{H}$:
\begin{equation}
F = \{V_j: \emptyset \subset V_1 \subset V_2 \cdots V_{n_d-1} \subset  V_{n_d} = \mathcal{H}\}.
\end{equation}
In our case, let $V_\alpha$ be the direct sum of the eigenspaces belonging to the $\alpha$ largest elements in $\Lambda^d$. If all eigenvalues are distinct, the flag is \emph{complete}, \emph{i.e.} the dimension of each consecutive subspace differs by one, so that $n_d = n$. Let $P_\alpha^F$ be the projector associated with the $\alpha$th eigenspace, so that $P_1^F\oplus\cdots\oplus P_\alpha^F$ projects onto $V_\alpha$ of the flag $F$. Henceforth, we will identify a flag with a tuple of orthogonal projectors $P:= (P_1,\dots,P_{n_d})$, and we will use $\pi=(\pi_1,\dots,\pi_n)$ to denote a complete flag \footnote{This is a slight abuse of the term, as the flag is the family of subspaces, not the tuple of projectors. But there is a clear one-to-one correspondence, and it will be easier to work with the projectors.}.   

Let $\mathcal{F}$ be the set of all flags on $\mathcal{H}$. Let $\xi$ be a composition\footnote{A composition of $n$ is an ordered set of positive integers that sum to $n$} of $n$. Then $\mathcal{F}^\xi$ is the set of all flags such that $dim(P^F_\alpha)$ equals the $\alpha$th element of $\xi$ (in our case, this will be $m_\alpha$). We want to consider functions $P(t)$ taking values in $\mathcal{F}$ that are differentiable, but we need to be careful as to what this means in the vicinity of eigenvalue crossings, since $n_d$ is not constant. In fact, we will distinguish between two notions of differentiability. Let $P_{\alpha,t_0}(t)$ be the \emph{total projector} of the $\alpha$th eigenspace in a neighborhood around time $t=t_0$: $P_{\alpha,t_0}(t_0) = P_{\alpha}(t_0)$ and the projectors in $P_{\alpha,t_0}(t\ne t_0)$ are the sums of projectors in $P_{\alpha'}(t)$ corresponding to eigenvalues that cross at $t_0$. We say that $P(t)$ is \emph{weakly differentiable} at time $t=t_0$ if each total projector $P_{\alpha,t_0}(t)$ is differentiable there. This does not imply the lower-dimensional projectors of the crossing eigenvalues are differentiable in that neighborhood (see \cite{KatoBook}, section II.3 for a counter-example). We say that a flag $P(t)$ is \textit{strongly differentiable} if it can be built out of a complete flag $\pi (t)$ whose elements are differentiable. 

Define $(A)_{\alpha\beta}^P:= P_\alpha A P_\beta$ for any operator $A$ (we will drop the superscript $P$ when it is understood). We can now write down a theorem about the decomposition of $\rho$ into its eigenvalues and eigenvectors:

\begin{theorem}\label{thm1}
Suppose $\rho(t)$ obeys the Lindblad equation (\ref{lindblad}) where the Hamiltonian operator is continuous. Then there is a differentiable $\Lambda(t)$ and weakly differentiable flag $P(t)$ such that $\rho(t) = \sum_\alpha \lambda_\alpha^d P_\alpha$. At any sharp eigenvalue crossing, $P(t)$ is strongly differentiable. The derivative of a total projector is given by the formula:
\begin{equation}
\frac{d}{dt}P_{\alpha,t_0} = \sum_{\beta\ne\alpha}\frac{(\rho ')_{\alpha\beta} + (\rho ')_{\beta\alpha}}{\lambda_\alpha^d - \lambda_\beta^d}. \label{projderiv}
\end{equation}
The eigenvalue derivatives corresponding to $\lambda_\alpha^d$ are the eigenvalues of $(\rho^\prime(t))^{P(t)}_{\alpha\alpha}$. 
\end{theorem}

\begin{proof}
Reference \cite{KatoBook}, specifically Theorem 5.4 from chapter two therein, covers much of this theorem. It states the differentiability of $\Lambda$ and differentiability of the total projectors, as well the formulas for their derivatives. The two stipulations are that (i) $\rho(t)$ is differentiable, which is true as we have required $H(t)$ to be continuous, and (ii) all $\lambda_j$ are semi-simple, which is of course true for all Hermitian operators. The formula for the eigenvalue derivative is also provided in this reference. The formula for the eigenprojector derivative is given as $\frac{d}{dt}P_{\alpha,t_0} = - P_{\alpha,t_0}\rho' S_\alpha -S_\alpha\rho' P_{\alpha,t_0}$, where $S_\alpha = \sum_{\beta\ne\alpha} \frac{P_\beta}{\lambda_\alpha^d-\lambda_\beta^d}$. Our formula clearly follows.

All that is left to prove is strong differentiability at a sharp crossing, which requires some care. First note the crossing must be isolated: if the crossing is at $t=t_0$, there is a neighborhood $T_0$ on which all eigenvalues are distinct for $t\ne t_0$. We must find a complete set of differentiable one-dimensional orthogonal projectors $\bar{\pi}_l(t)$ on $T_0$ that sum to the relevant higher-dimensional projectors at $t=t_0$, and also obey the formula for projector derivatives given in the theorem. Let $C_\alpha(t_0)$ be the subset of indices $1$ through $n$ corresponding to the eigenvalues that equal $\lambda_\alpha^d$ at $t_0$. Now for $l\in C_\alpha (t_0)$, define the $\bar{\pi}_l(t_0)$'s to be the eigenprojectors of $(\rho'(t_0))_{\alpha\alpha}$, and $\mu_l$ the corresponding eigenvalues. Note that $\lambda_l'(t_0) = \mu_l$. Since the eigenvalue crossing is sharp, all $\mu_l$ for $l\in C_\alpha$ are distinct, and therefore $\bar{\pi}_l(t_0)$ is well-defined. Moreover, $\sum_{l\in C_\alpha(t_0)}\bar{\pi}_l (t_0) = P_\alpha(t_0)$. 

For $t \ne t_0$ and $l\in C_\alpha$, define $\bar{\pi}_l(t)$ to be the solution of an ODE: 
\begin{align}
\frac{d}{dt}\bar{\pi}_l = \sum_{m=1, m\ne l}^n \frac{(\rho ')_{lm} + (\rho ')_{ml}}{\lambda_l - \lambda_m}. \label{barpilim}
\end{align}
This ODE does not appear to be well-defined at $t=t_0$, but we claim the limit of the RHS exists as $t\rightarrow t_0$, and we define the ODE to be this limit at $t=t_0$.

To prove our claim, we must show that if $m\in C_\alpha$, its corresponding numerator must approach zero just as fast as $\lambda_m-\lambda_l$. Because the eigenvalue crossing is sharp, the denominator goes to zero linearly: it is $(\mu_m-\mu_l)\delta t +o(\delta t)$ for small $\delta t$, and $\mu_m \ne \mu_l$ by assumption. The numerator goes to zero because, if we write $\rho(t) = \rho(t_0)+\rho'(t_0)\delta t + o(\delta t)$, and $\tau := [-iH,\rho '] +\mathcal{L}_D(\rho')$:
\begin{align}
(\rho'(t))_{lm} &= (\rho'(t_0))_{lm} + (\tau(t_0))_{lm} \delta t  + o(\delta t) \nonumber \\
&= (\tau(t_0))_{lm} \delta t  + o(\delta t),  \label{limit}
\end{align}
where we have substituted the expression for $\rho(t)$ into the Lindblad equation and applied the projectors $\bar{\pi}_l$ and $\bar{\pi}_m$. The first term is zero as we have constructed $\bar{\pi}_l$ and $\bar{\pi}_m$ to be the eigenprojectors of $(\rho'(t_0))_{\alpha\beta}$ and $l\ne m$. 

The ODE is then well defined. It is bounded and thus Lipschitz, so the $\bar{\pi}_l$ have a well-defined solution. By construction, they obey the formula (\ref{projderiv}) for $t\ne t_0$, and all that remains to show is that it obeys the formula at $t=t_0$. In other words, we must sum the ODE's for all $l\in C_\alpha$:
\begin{align}
\frac{d}{dt}P_{\alpha ,t_0} &= \frac{d}{dt}\sum_{l\in C_\alpha}\bar{\pi}_l \\
&= \lim_{t\rightarrow t_0}  \sum_{l\in C_\alpha}\sum_{m=1, m\ne l}^n \frac{(\rho ')_{lm} + (\rho ' )_{ml}}{\lambda_l - \lambda_m} \\
&= \lim_{t\rightarrow t_0} (\sum_{l,m\in C_\alpha, l\ne m} + \sum_{l\in C_\alpha, m\notin C_\alpha}) \frac{(\rho ')_{lm} + (\rho ' )_{ml}}{\lambda_l - \lambda_m} \\
&= \lim_{t\rightarrow t_0}  \sum_{l\in C_\alpha, m\notin C_\alpha} \frac{(\rho ')_{lm} + (\rho ' )_{ml}}{\lambda_l - \lambda_m} \\
&=  \sum_{\beta=1,\beta\ne \alpha}^{n_d} \frac{(\rho ')_{\alpha\beta} + (\rho ' )_{\beta\alpha}}{\lambda^d_\alpha - \lambda^d_\beta}. 
\end{align}
In the second-to-last line, the first summation vanishes because of cancellation of terms with swapped indices. Thus we have proven strong differentiability of the flag at sharp eigenvalue crossings.   
\end{proof}

\begin{remark}
Strong differentiability often holds for non-sharp crossings as well. The limit of (\ref{barpilim}) is well-behaved as long as, for every (isolated) crossing pair, there is a higher-order derivative of $\rho$ at which the corresponding derivative eigenvalues differ. If  $\rho$ is analytic, strong differentiability holds: either two crossing eigenvalues have some order of derivative at which they can be resolved, or they are identical over some neighborhood. The pathological counter-example mentioned in \cite{KatoBook}, for example, involves a smooth but non-analytic operator. In this case, two eigenvalues can have identical derivatives at all orders, and yet the crossing is isolated. In our case, we only require $\rho$ to be differentiable, so higher-order derivatives may not exist.
\end{remark}

\begin{remark}
If the Hamiltonian is piecewise-continuous instead of continuous, we can easily modify the theorem as long as right- and left-sided limits exist at the discontinuities. If such limits exist, the corresponding one-sided derivatives of $\rho$ exist, and so there is no problem. If, on the other hand, $||H(t)||\rightarrow\infty$ for $t\rightarrow t_0\pm$, the differentiability properties of the eigenvalues and projectors clearly do not hold.
\end{remark}

Now let us write down formulas for the derivatives of $\Lambda$ and $\pi$. We have the following proposition:
\begin{proposition}
If $\rho (t)$ obeys the Lindblad equation, and $( \Lambda(t),\pi (t))$ is a differentiable decomposition of $\rho(t)$, define $w_{ij}^\pi = \sum_{k=1}^N \textrm{Tr}( \pi_i L_k \pi_j L_k^\dagger )$. Then:
\begin{align}
\frac{d}{dt}\Lambda (t) = \Omega^{\pi(t)}  \Lambda(t) \label{LambODE},
\end{align}
where $\Omega^{\pi}$ is an $n$-by-$n$ matrix with:
\begin{equation}
\Omega^{\pi} = \left\{ \begin{array}{cc}
w_{ij}^{\pi}, & i \ne j \\
- \sum_{l\ne j} w_{lj}^{\pi}, &  i = j . 
\end{array} \right.
\end{equation}
This formula holds for sharp eigenvalue crossings.
\end{proposition}

\begin{proof}
We know that if $\Lambda(t)$ is differentiable, $\frac{d}{dt}\Lambda$ is given by the eigenvalues of the operators $(\rho')_{\alpha\alpha}^F$. Since we know the elements of $\pi$ are their eigenprojectors, we can retrieve their eigenvalues by tracing over the one-dimensional projections of $\rho'$.
\begin{align}
\frac{d}{dt}\Lambda_j  &=  \textrm{Tr}(\pi_j \rho' \pi_j) \\
&= \sum_{k=1}^n \textrm{Tr}(\pi_j [-iH,\lambda_k\pi_k]\pi_j + \pi_j\mathcal{L}_D (\lambda_k\pi_k)\pi_j)\\
&= \sum_{k=1}^n \sum_{l=1}^N \textrm{Tr}(\pi_j L_l \lambda_k\pi_k L_l^\dagger \pi_j - \frac{1}{2}\pi_j\{L_l^\dagger L_l, \lambda_k\pi_k \}\pi_j)\\
&= \sum_{k=1}^n \sum_{l=1}^N \textrm{Tr}(\lambda_k \pi_j L_l \pi_k L_l^\dagger \pi_j - \lambda_j \pi_j L_l^\dagger \pi_k L_l \pi_j )\\
&= \sum_{k=1}^n \lambda_k w_{jk}^{\pi} - \lambda_j w_{kj}^{\pi}\\
&= \sum_{k=1}^n \Omega^{\pi}_{jk}\Lambda_k.
\end{align}
We have made use of the identities $\pi_j\pi_k = \delta_{jk}$ and $\sum_{k=1}^n \pi_k = I_n$. 
\end{proof}

\begin{corollary}
$\Omega^{\pi}$ is rank-deficient. On the projected simplex, we have the formula for $x(t)\in\bar{\mathcal{T}}$:
\begin{equation}
\frac{d}{dt}x(t) = b^{\pi(t)} + A^{\pi(t)}x(t), \label{xeq}
\end{equation}
where $b^{\pi} = \Pi \Omega^{\pi} \iota$ and $A^{\pi} =  \Pi \Omega^{\pi} \Pi^T$. 
\end{corollary}

\begin{proof}
$\Omega^{\pi}$ must be rank-deficient because its column-sums are zero, which is a reflection of the fact that the element-sum of $\Lambda$ must be one. The ODE is obtained by substituting $\Lambda = \iota + \Pi^T x$ into the ODE in the proposition, and then multiplying by $\Pi$.
\end{proof}

\section{The Projected Control System}

We have decomposed the Lindblad system into its spectrum and flag, and now we want to define a new control system. Let us clarify the distinction between the old and new control systems:

\begin{definition}
The \emph{$\rho$-control system} is the Lindblad equation (\ref{lindblad}), a complete set of control Hamiltonians $\{H_i\}$ that span the Lie algebra $\mathfrak{su}(n)$, and the control functions $u_i(t)$ that are piecewise-continuous, real-valued and unbounded.
\end{definition}

For a flag $\pi$ or $P$ and eigenvalue vector $\Lambda$, define the following maps:
\begin{align}
\mathcal{M}_{ij}(\Lambda,\pi) &= \pi_i \mathcal{L}_D(\sum_l\lambda_l\pi_l)\pi_j \\
\mathcal{M}_{\alpha\beta}(\Lambda,P) &= P_\alpha \mathcal{L}_D(\sum_\gamma\lambda_\gamma^dP_\gamma)P_\beta.
\end{align}

Let $\mathbb{F}\subset\mathcal{F}$ be the set of complete flags. Then:

\begin{definition}\label{LCS}
The \emph{$\Lambda$-control system} is the linear ODE (\ref{xeq}), together with control flags $\pi(t)$ on the control-set $\mathbb{F}$. We consider only functions $\pi(t)$ that are piecewise-differentiable. Additionally, the control functions must meet the following two conditions: 
\begin{enumerate}
\item At any crossing $\lambda_i=\lambda_j$, there is a neighborhood and $C > 0$ such that $|| \mathcal{M}_{ij}(\Lambda(t),\pi(t)) || \le C||\lambda_i(t)-\lambda_j(t) ||$.
\item $\pi(t)$ must satisfy an initial and a final condition: $\pi(t_i)=\pi_i$ and $\pi(t_f)=\pi_f$.
\end{enumerate}
\end{definition}

The first condition is essentially the requirement that $\pi(t)$ always diagonalizes $(\rho)_{\alpha\alpha}$ at crossings, and that it is sufficiently well-behaved in the vicinity of the crossing that a bounded Hamiltonian can be recovered. Let $\mathbb{F}^\Lambda$ denote the set of $\pi$ that satisfy $\mathcal{M}_{ij}(\Lambda, \pi) = 0$ for $\lambda_i = \lambda_j$. $\mathbb{F}^\Lambda \ne \mathbb{F}$ at crossings, so the control set shrinks: we are free to choose the  projectors $P_\alpha$, but not their diagonalizations. The dimension of the control set is $n^2-\sum_\alpha m_\alpha^2$. When all eigenvalues are simple, this dimension is $n^2-n$. Conversely, at the completely mixed state where $\rho =\frac{1}{n}I_n$, $\Lambda = \iota$ and the control set is a singleton. 

The second condition above is imposed since we typically have an initial and target density matrix in mind, each with their own flags that we may not choose. Note that both conditions can be dropped if we are willing to settle for approximate controllability: that is, if it suffices that our final $\rho$ is arbitrarily close to our target $\rho$. We will expand on this shortly.

We can now write down a formula for the Hamiltonian:

\begin{proposition}
Given a trajectory $\Lambda(t)$ and controls $\pi(t)$ in the $\Lambda$-control system, we can recover the density operator $\rho(t)=\sum_j \lambda_j(t)\pi_j(t)$ using the following Hamiltonian:
\begin{align}
H^\pi(t) &= i \Big( -\sum_{j=1}^{n} \pi_j(t) \pi_j^\prime(t)  + \sum_{\substack{\alpha,\beta=1 \\ \alpha\ne \beta}}^{n_d}\left.\frac{\mathcal{M}_{\alpha\beta}(\Lambda,P)}{\lambda^d_\alpha(t)-\lambda^d_\beta(t)}  \right), \label{ham}
\end{align}
where $P_\alpha = \sum_{j\in C_\alpha}\pi_j$. This Hamiltonian is piecewise-continuous.
\end{proposition}

\begin{proof}
Firstly, note that the piecewise-continuity follows from condition one in the definition of the $\Lambda$-control system. If we write the two terms of the Hamiltonian $H^\pi = H_A^\pi + H_B^\pi$, it is clear that $H_A^\pi$ is piecewise-continuous due to the piecewise-differentiability of $\pi$. $H_B^\pi$ is piecewise-differentiable because the numerator and denominator are, and condition one demands the numerator always approaches zero at least as fast as the denominator.

We must now show that our re-constructed $\rho(t)$ and $H^{\pi}(t)$ obey the Lindblad equation (\ref{lindblad}), which amounts to:
\begin{equation}
\sum_{j=1}^{n} \left( \lambda_j^{\prime} \pi_j + \lambda_j \pi_j^{\prime} \right)=
\sum_{j=1}^{n} \left([-iH^\pi,\lambda_j \pi_j] + \mathcal{L}_D(\lambda_j \pi_j)\right). \label{decomplind}\end{equation}
We claim that $[-iH^\pi_A,\sum_j \lambda_j\pi_j] = \sum_j \lambda_j\pi_j^\prime$ and that $[-iH^\pi_B,\sum_j \lambda_j\pi_j]+\sum_j \mathcal{L}_D(\lambda_j\pi_j) = \sum_j \lambda_j^\prime\pi_j$, which if true would prove the proposition. 

For the first part of the claim: 
\begin{align}
&[-iH^\pi_A,\sum_j \lambda_j\pi_j] = -\sum_{j, k=1}^n [\pi_k \pi_k^\prime, \lambda_j\pi_j] \\
&= - \sum_{j, k=1}^n \lambda_j \pi_k \pi_k^\prime \pi_j + \sum_{j=1}^{n} \lambda_j \pi_j \pi_j^\prime \\
&= - \sum_{j, k=1}^n \lambda_j (\pi_k^\prime - \pi_k^\prime \pi_k) \pi_j + \sum_{j=1}^{n} \lambda_j (\pi_j^\prime -  \pi_j^\prime \pi_j) \\
&= - \sum_{j, k=1}^n  (\pi_k^\prime \lambda_j\pi_j - \lambda_k\pi_k^\prime \pi_k)  + \sum_{j=1}^{n} \lambda_j (\pi_j^\prime -  \pi_j^\prime \pi_j) \\
&= \sum_{j=1}^{n} \lambda_j \pi_j^\prime, \end{align}
where we have used the identities  $\sum_{k} \pi^\prime_k = 0$, and $\pi_j^\prime \pi_j + \pi_j \pi_j^\prime = \pi_j^\prime$. 

For the second part of the claim:
\begin{align}
&[-iH^\pi_B,\sum_{j=1}^n \lambda_j\pi_j]+\sum_{j=1}^n \mathcal{L}_D(\lambda_j\pi_j) = \sum_{j=1}^n \Big( \mathcal{L}_D(\lambda_j\pi_j) \nonumber \\
&\hspace{1.5cm} +\sum_{\substack{\alpha,\beta=1 \\ \alpha\ne \beta}}^{n_d} \sum_{\gamma=1}^{n_d}\left. \frac{[P_\alpha \mathcal{L}_D(\lambda_\gamma^d P_\gamma) P_\beta, \lambda_j\pi_j]}{\lambda^d_\alpha-\lambda^d_\beta} \right) \\
&= \sum_{j=1}^n \mathcal{L}_D(\lambda_j\pi_j) - \sum_{\substack{\alpha,\beta=1 \\ \alpha\ne \beta}}^{n_d} \sum_{\gamma=1}^{n_d} P_\alpha \mathcal{L}_D(\lambda_\gamma^d P_\gamma) P_\beta  \\
&= \sum_{\alpha =1}^{n_d} P_\alpha \mathcal{L}_D(\rho) P_\alpha  = \sum_{\alpha =1}^{n_d} P_\alpha\left( [-iH^\pi,\rho] +\mathcal{L}_D(\rho) \right) P_\alpha \\
&=\sum_{\alpha =1}^{n_d}P_\alpha\frac{d\rho}{dt}P_\alpha = \sum_{j=1}^n \lambda_j^\prime \pi_j.
\end{align}
So our construction obeys the Lindblad  equation. 
\end{proof}

Note that the constructed $H^\pi(t)$ may become very large if two eigenvalues become very close. If the eigenvalues actually cross however, the Hamiltonian is well-behaved. There are only certain $\pi$ that allow an eigenvalue crossing, and trying to approach a crossing with an illegal $\pi$ requires an infinite energy cost. Note that orbits with repeated eigenvalues must fall on the boundary of $\mathcal{T}_Q$, so if we only require that we steer arbitrarily close to such an orbit, we can ignore the first condition, since nearby points are in the interior where the condition does not apply. 

We now explore the implications of eliminating the second condition. If we construct a trajectory $(\Lambda(t),\pi(t))$ with the desired initial and final $\Lambda$, but with an undesired initial and final $\pi$, we can book-end the trajectory with fast unitary transformations. Say we have initial and target density operators $\rho_i$ and $\rho_T$. We are able to construct $\Lambda(t)$ and $\pi(t)$ on the interval $[0,T]$ that brings $\rho_1$ to $\rho_2$, where there are skew-symmetric matrices $-ih_i$ and $-ih_T$ such that $\rho_1 = e^{-ih_i}\rho_ie^{ih_i}$ and $\rho_2 = e^{-ih_T}\rho_Te^{ih_T}$. Then we can construct the following motion on the interval $[-\Delta, T+\Delta]$: 
\begin{align}
t\in [-\Delta,0] \left\{
\begin{array}{c} 
  \rho_i\rightarrow\bar{\rho}_1 \\
H(t) = h_i/\Delta 
\end{array} \right.  \\
t\in [0,T] \left\{
\begin{array}{c} 
  \bar{\rho}_1\rightarrow\bar{\rho}_2 \\
H(t) = H^\pi(t) 
\end{array} \right.  \\
t\in [T,T+\Delta] \left\{
\begin{array}{c} 
  \bar{\rho}_2\rightarrow\rho_f \\
H(t) = h_T/\Delta. 
\end{array} \right.  
\end{align}
Let $\rho_a(t)$ denote our ideal trajectory $\rho_i\rightarrow\rho_1\rightarrow\rho_2\rightarrow\rho_T$ and $\rho_b(t)$ the actual trajectory $\rho_i\rightarrow\bar{\rho}_1\rightarrow\bar{\rho}_2\rightarrow\rho_f$. To measure distance between density operators, we will use the trace distance\footnote{See \cite{BengtssonZyczkowskiBook} for other distance measures for density matrices.}: 
\begin{align}
d(\rho_a,\rho_b) = \frac{1}{2}Tr(\sqrt{(\rho_a-\rho_b)^2}) = \frac{1}{2}\sum_{k=1}^n |\lambda^\delta_k|,
\end{align}
where $\lambda^\delta_k$ are the (real) eigenvalues of $\rho_a-\rho_b$. 

\begin{proposition}
Exact controllability in the $\Lambda$-control system without condition (2) implies approximate controllability in the $\rho$-control system. That is, if there is a $\pi(t)$ on $[0,T]$ that brings $\Lambda_{\rho_i}$ to $\Lambda_{\rho_T}$, then there is a Hamiltonian $H(t)$ on $[-\Delta, T+\Delta]$ that brings $\rho_i$ to $\rho_f$ such that $d(\rho_f,\rho_T)\le C \Delta $, where the constant $C$ is universal for all initial and final density operators.  
\end{proposition}

\begin{proof}
To begin, we note that the time-derivative of the distance is $d^\prime(\rho_a,\rho_b) = \sum_{k \in C^\delta(t)}\lambda^{\delta\prime}_k$, where $C^\delta(t)$ is the subset of indices such that $\lambda_k^\delta >0$. If one or more eigenvalues are zero with non-zero derivative, the metric has different left- and right-side derivatives. In this case, define $C^\delta(t-)$ to include  indices for zero and decreasing eigenvalues, and $C^\delta(t+)$ to include indices for zero and increasing eigenvalues. We know that the eigenvalues are differentiable, since theorem \ref{thm1} can be applied with minimal modification to $\rho_a-\rho_b$. 

Now for the first part of the trajectory:
\begin{align}
d(\rho_1,\bar{\rho}_1) &\le \Delta \cdot \sup_{-\Delta\le t\le 0} |d^\prime(\rho_a(t),\rho_b(t)) |\\
&\le \frac{\Delta}{2} \cdot \sup_{-\Delta\le t\le 0} \sum_{k=1}^n | \lambda^{\delta\prime}_k(t) |\\
&= \frac{\Delta}{2} \cdot \sup_{-\Delta\le t\le 0} \sum_{k=1}^n | \mu^{\delta}_k(t) |,
\end{align}
where the $\mu^\delta_k$ are eigenvalues of $(\rho_a^\prime-\rho_b^\prime)$ projected onto its different eigenspaces. Now $(\rho_a^\prime-\rho_b^\prime) = [-ih_i/\Delta,\rho_a-\rho_b] + \mathcal{L}_D(\rho_b)$. The Hamiltonian piece projected onto its eigenspaces vanishes, so we are left with only the dissipative piece. It follows that $\mu^\delta_k \le \sup_{-\Delta\le t\le 0} |\mathcal{L}_D(\rho_b(t))|\le 2\sum_{m=1}^N|L_m|^2$. So $d(\rho_1,\bar{\rho}_1) \le n\Delta\sum_{m}|L_m|^2$.

The middle piece of the trajectory causes no problems, since both $\rho_a$ and $\rho_b$ experience the same dynamics, and the Lindblad equation is known to be contractive \cite{Lindblad76}. We can adapt equation (\ref{LambODE}) for $(\rho_a-\rho_b)$ instead of $\rho$, where $\Lambda^\delta$ and $\pi^\delta$ replace $\Lambda$ and $\pi$ (this can be done since the positive semi-definiteness is not invoked in the proof). On the interval $[0,T]$, we have:
\begin{align}
d^\prime (\rho_a,\rho_b) &= \sum_{k\in C^\delta}\lambda_k^{\delta\prime} = \sum_{k\in C^\delta}\sum_{l=1}^n \Omega^{\pi^\delta}_{kl}\lambda^\delta_l \\
&= \left(\sum_{k, l\in C^\delta} + \sum_{k\in C^\delta, l\not\in C^\delta}\right)\Omega^{\pi^\delta}_{kl}\lambda^\delta_l \\
&= - \sum_{k\not\in C^\delta, l\in C^\delta}w^\delta_{kl}\lambda^\delta_l
 + \sum_{k\in C^\delta, l\not\in C^\delta}w^\delta_{kl}\lambda^\delta_l \\
&= - \sum_{k\not\in C^\delta, l\in C^\delta}w^\delta_{kl}|\lambda^\delta_l|
 - \sum_{k\in C^\delta, l\not\in C^\delta}w^\delta_{kl}|\lambda^\delta_l| \\
& \le 0,
\end{align} 
where in the third line, first sum, we have used the fact that the column-sums of $\Omega^{\pi^\delta}$ are zero. 

So $|\rho_2-\bar{\rho}_2| \le |\rho_1-\bar{\rho}_1|$. To finish, we have:
\begin{align}
d(\rho_f,\rho_T) &\le d(\rho_2,\bar{\rho}_2) + \Delta \cdot \sup_{T \le t\le T+\Delta} |d^\prime(\rho_a(t),\rho_b(t)) | \\
& \le 2n\Delta\sum_{m=1}^N|L_m|^2.
\end{align}
The multiplicative constant $2n\sum_{m=1}^N|L_m|^2$ is independent of $\rho_i$ and $\rho_f$. 
\end{proof}

\begin{corollary}
If we expand the $\Lambda$-control system to allow piecewise-differentiable $\pi(t)$ with a finite number of discontinuities, the final density operator corresponding to the final $\Lambda$ can be reached within an arbitrarily small error.
\end{corollary}

\begin{proof}
This is merely an extension of the previous lemma, where instead of book-ending one continuous trajectory with fast unitary transformations, we are intersplicing a finite number of fast unitary transformations at the discontinuities. 
\end{proof}

While the conditions in the definition of the $\Lambda$-control system are necessary for planning trajectories in $\rho$-space and their corresponding Hamiltonians, they can be disregarded when analyzing controllability. This will be made clearer in the next section; for now, we define the following control system:

\begin{definition}
The \emph{unconstrained $\Lambda$-control system} is the linear ODE (\ref{xeq}), together with a piecewise-differentiable control flag $\pi(t)$, with a finite number of possible discontinuities.   
\end{definition} 

Because the control set of the $\Lambda$-control system is a non-Euclidean manifold, it is not trivial to use standard control-theoretic results for the projected system. However, if we view the elements $w_{ij}^\pi$ as controls, we are left with a bi-linear control system, since $\Omega^\pi$ is linear in these elements. Define the map $w: \mathbb{F}\rightarrow \mathbb{R}^{n^2-n}$ that sends $\pi$ to the corresponding vector of $w_{jk}^\pi$. Note that $w(\mathbb{F})$ is a closed and bounded set in $\mathbb{R}^{n^2-n}$. Also define $\Omega(w)$, $w\in\mathbb{R}^{n^2-n}$ to be the matrix with off-diagonal elements equal to $w_{jk}$ and diagonal elements equal to $-\sum_{l \ne j} w_{lk}$. Define the following control system, which is the unconstrained $\Lambda$-control system with a transformation:

\begin{definition}
The $w$-control system is the bi-linear ODE $\frac{d}{dt}\Lambda = \Omega(w)\Lambda$ on $\mathcal{T}$. The control set is $w(\mathbb{F})$ and control functions must be piecewise-differentiable, with a finite number of discontinuities.
\end{definition}

The derivatives of $w$ are, where $h\in T_\pi\mathbb{F}\subset \mathfrak{su}(n)$:
\begin{align}
dw_{jk}(\pi)\cdot h  &= \sum_{l=1}^N\pi_j [L_l, h]\pi_k L_l^\dagger\pi_j +\pi_j L_l\pi_k[L_l^\dagger,h]\pi_j \\
w'_{jk}(t) &= \sum_{l=1}^N\pi_j(t) [L_l, \pi'(t)]\pi_k(t) L_l^\dagger\pi_j(t)\nonumber \\
 &+\pi_j(t) L_l\pi(t)_k[L_l^\dagger,\pi'(t)]\pi_j(t). \label{w-deriv}
\end{align}
Since $w(t)$ is confined to $w(\mathbb{F})$, $w'$ must be in the image of $dw(\pi)$, and therefore we can recover $\pi^\prime(t)$ and therefore differentiable $\pi(t)$ from $w^\prime(t)$ and $w(t)$. It follows that the $w$-control system is equivalent to the unconstrained $\Lambda$-system. The difficulty in analyzing the $w$-control system is understanding the structure of the control set $w(\mathbb{F})$. 

\section{Local Controllability Analysis}

In the remainder of this paper, we wish to examine the controllability of the $\Lambda$-control system. We will restrict ourselves to local controllability, as this simplifies the analysis somewhat:

\begin{definition}
A control system is \emph{locally controllable (LC)} \cite{SontagBook} in time $T$ at a point $p$ if for every neighborhood $V$ of $p$, $V$ contains another neighborhood $W$ such that $\forall y,z\in W$, $y$ can be controlled to $z$ in time $T$. The system is \emph{strongly locally controllable (SLC)} if a $W$ can be found for any $V$ such that $\forall y,z\in W$, $y$ can be controlled to $z$ without leaving $W$. 
\end{definition}

In plain terms, local controllability guarantees a trajectory between two local points, while strong local controllability demands this trajectory also be local. We will give a sufficient condition for SLC in both the unconstrained and constrained $\Lambda$-control system. First define $\mathcal{V}_u(\Lambda)=\{\Omega^\pi\Lambda: \pi\in\mathbb{F}\}$ and  $\mathcal{V}_c(\Lambda)=\{\Omega^\pi\Lambda: \pi\in\mathbb{F}^\Lambda\}$. These are the possible tangent vectors $\frac{d}{dt}\Lambda$ available at $\Lambda$ for the unconstrained and constrained systems. Here, $int$ denotes ``interior" and $co$ ``convex hull":

\begin{proposition}
If $0\in \textrm{int co } \mathcal{V}_u(\Lambda)$, then both the unconstrained and constrained $\Lambda$-systems are SLC at $\Lambda$. If $0\not\in \textrm{co } \mathcal{V}_u(\Lambda)$, neither are LC at $\Lambda$. \label{SLC}
\end{proposition}

\begin{proof}
The first part is an application of Lemma 3.8.5 and its corollary from \cite{SontagBook}, which states that if $0$ lies in the interior of the convex hull of the set of available tangent vectors, then the system is SLC. The wrinkle we must deal with is showing that the SLC extends to the constrained system, despite the smaller control set. 

For the constrained system, we claim that $\textrm{co }\mathcal{V}_u(\Lambda) = \textrm{co }\mathcal{V}_c(\Lambda)$, which if true yields the desired result. Our claim follows from the Schur-Horn theorem \cite{Schur}\cite{Horn}, which states that for any Hermitian operator $A$, $\{diag(UAU^\dagger): U\in U(n)\} = \textrm{co }\{\sigma.\Gamma_A: \sigma\in S_n\} $. Here $diag()$ denotes the vector of diagonal elements, $\sigma .\Gamma$ denotes $\Gamma$ with elements permuted with $\sigma\in S_n$, and $\Gamma_A$ denotes the vector of eigenvalues of $A$. This can be extended to direct sums: for any set of Hermitian operators $A_\alpha$, $\{\bigoplus_\alpha diag(U_\alpha A_\alpha U_\alpha^\dagger): U_\alpha \in U(n_\alpha)\} = co(\{\bigoplus_
\alpha \sigma.\Gamma_{A_\alpha}: \sigma\in S_{n_\alpha}\})$. In our case we use $A_\alpha = (\rho^\prime)_{\alpha\alpha}$. Then we have:
\begin{align}
\textrm{co }\mathcal{V}_c(\Lambda) &= \textrm{co }\{\bigoplus_\alpha \sigma_\alpha.\Gamma_{\rho'_{\alpha\alpha}}: \sigma_\alpha\in S_{m_\alpha}, P_\alpha \in\mathbb{F}^\Lambda \} \\
= \{\bigoplus_\alpha & diag( U_\alpha P_\alpha \rho'P_\alpha U_\alpha^\dagger): U_\alpha\in U(m_\alpha), P_\alpha \in\mathbb{F}^\Lambda \}  \\
&= \{\bigoplus_j diag( \pi_j\rho'\pi_j): \pi \in\mathbb{F} \}  \\
&= \textrm{co }\{ \sigma.\Gamma_{\rho'}: \sigma\in S_n \}  = \textrm{co }\mathcal{V}_u(\Lambda).
\end{align}
In the second and fourth lines, we apply the Schur-Horn theorem. In the third line, we recognize the set of all diagonal vectors of $\rho'_{\alpha\alpha}$ is equal to the set of all possible $\bigoplus_{j\in C_\alpha} \pi_j\rho'\pi_j$ for $\sum_{j\in C_\alpha}\pi_j = P_\alpha$.  

To show the second part of the proposition, note that $\textrm{co }\mathcal{V}_u(\Lambda)$ is compact, since $w(\mathbb{F})$ and thus $\mathcal{V}_u(\Lambda)$ is compact, and the convex hull of a compact set in $\mathbb{R}^n$ is compact. Suppose at some $\Lambda=\Lambda_0$, $0\not\in\textrm{int co }\mathcal{V}_u(\Lambda_0)$. Due to the compactness and convexity, there is a unique point $v_m \in \partial \textrm{ co }\mathcal{V}_u(\Lambda_0)\subset\mathbb{R}^n$ with minimal magnitude, and this fixes a hyperplane passing through $\Lambda_0$ that is orthogonal to $v_m$. The magnitude of this vector as $\Lambda$ varies cannot vary more than $C_\Omega|\delta \Lambda|$, where $C_\Omega = \sup |\Omega(w(\mathbb{F}))|$. Due to compactness, there is also a point $v_M$, not necessarily unique, of maximal magnitude. If we define $\tau = \frac{|v_m|}{2|v_M| C_\Omega}$, then $R_{\tau}(\Lambda_0)$ falls entirely on one side of the hyperplane and thus cannot contain zero. This is because:
\begin{align}
\left(\Lambda(\tau)-\Lambda_0\right) &\cdot \frac{v_m}{|v_m|} \ge \tau \inf_{t\in[0,\tau]} (\Lambda^\prime(t) \cdot \frac{v_m}{|v_m|}) \\
&\ge \tau (|v_m| - C_\Omega \sup_{t\in [0,\tau]}|\Lambda^\prime(t)|\tau  ) \\
&\ge \tau (|v_m| - C_\Omega |v_M |\tau  ) = \frac{1}{2} \tau |v_m| > 0.
\end{align}
It follows that $LC$ does not hold at $\Lambda_0$. 
\end{proof}

Analyzing the local controllability of the $\Lambda$-system requires studying $\mathcal{V}_u(\Lambda)$. For general $\Lambda\in\mathcal{T}$, this is difficult, but at the completely mixed state, its structure simplifies greatly, as it is the convex hull of a finite set of vectors:

\begin{proposition}
$\mathcal{V}_u(\iota) = \textrm{co }\{\sigma.\Gamma_{A_\iota}: \sigma\in S_n\}$, where $A_\iota$ is the operator $\sum_k[L_k,L_k^\dagger]$. 
\end{proposition} 

\begin{proof}
This is a consequence of the fact that when $\rho = \frac{1}{n}I$, $\rho^\prime = \frac{1}{n}A_\iota$. If one applies the Schur-Horn theorem, the proposition immediately follows.
\end{proof}

In general, $\mathcal{V}_u(\Lambda)$ is not the convex hull of a finite number of vectors, as it is at the completely mixed state. However, it does raise a tractable question: where does SLC hold for the $\Lambda$-control system when one is restricted to a finite control-set? To this end, we state a theorem (which is easier to state in terms of $x=\Pi\Lambda\in\mathbb{R}^{n-1}$ rather than $\Lambda\in\mathbb{R}^n$) about the region $\mathcal{A}\subset \bar{\mathcal{T}}$ where the necessary condition for SLC from proposition \ref{SLC} holds. It states that $\mathcal{A}$ is the image under a rational function of an $n$-simplex of parameters, and that the boundary $\partial \mathcal{A}$ is the image of the parameter-simplex's boundary. 

\begin{theorem}\label{SLCset}
Let $\{\pi^J: J = 1, \dots, n \}$ be a finite number of complete flags such that $A_J= \Pi\Omega^{\pi^J}\Pi^T$ is invertible $\forall J$. Also define $b_J = \Pi\Omega^{\pi^J}\iota$. Define the function:
\begin{align}
B: \mathcal{T}_s & \rightarrow \mathbb{R}^{n-1}\\
B(s) = - &\left(\sum_{J=1}^ns_J A_J \right)^{-1}\left(\sum_{J=1}^n s_Jb_J\right),
\end{align}
where $s=\langle s_1,\dots,s_n\rangle\in\mathcal{T}_s := \{s: s_J \ge 0, \sum_J s_J = 1\}$. If our control-set is $\{\pi^J\}$, then $\mathcal{A}=  \textrm{int }B(\mathcal{T}_s)$. Furthermore, $\partial B(\mathcal{T}_s) = B(\partial\mathcal{T}_s)$ and $\partial\mathcal{A} = B(\partial \mathcal{T}_s)$.
\end{theorem}

\begin{proof}
 The necessary condition for SLC is $0 \in \textrm{int co }\{b_J+A_Jx: J=1,\dots,n\}$. Either the $n$ points $b_J+A_Jx$ lie in a hyperplane, in which case the interior is empty, or they form an $n$-simplex. In the latter case, convexity means  the condition reduces to $0 = \sum_{J=1}^n s_J(b_J+A_Jx)$, $s\in\textrm{int }\mathcal{T}_s$. Re-writing we get:
\begin{align}
\sum_{J=1}^n s_JA_Jx  &= - \sum_{J=1}^n s_Jb_J \\
x &= - \left(\sum_{J=1}^n s_JA_J  \right)^{-1} \left(\sum_{J=1}^n s_Jb_J \right) = B(s). 
\end{align}
We can take the inverse because each $A_J$ is invertible, and since $\Omega$ is always negative semi-definite\footnote{$v^T\Omega v = \sum_{i<j} (w_{ij}+w_{ji})v_iv_j-w_{ij}v_i^2 - w_{ji}v_j^2 \le 0$ if $w_{ij}$, $w_{ji} \le 0$. }, each $A_J$ is also negative semi-definite.

In the exceptional case where the points are co-planar, we apply Carath\'{e}odory's Theorem \cite{Cara1911}, which says that any point in a convex hull of a set $P$ in an $m$-dimensional linear space must also lie in the convex hull of a set $P^\prime \subseteq P$ with at most $m$ elements. This means if the $b_J+A_Jx$ are co-planar, there is one we can eliminate without changing the convex hull. But this means one element of $s$ is zero, and this only occurs on $\partial\mathcal{T}_s$.
So the exceptional case only occurs if $x\in B(\partial\mathcal{T}_s)$. We will shortly show that $\partial B(\mathcal{T}_s) = B(\partial\mathcal{T}_s)$.  Therefore we must have $\mathcal{A} = \textrm{int }B(\mathcal{T}_s)$.

Next we show that $\partial B(\mathcal{T}_s) = B(\partial\mathcal{T}_s)$. There are three types of points on $\mathcal{T}_s$: boundary points, interior points that are critical points of $B$ and interior points that are regular points of $B$. Regular points must map to points in $\textrm{int }B(\mathcal{T}_s)$, due to the Inverse Function Theorem. To examine the interior critical points, write $A(s) = \sum_{J=1}^n s_JA_J$ and $b(s) = \sum_{J=1}^n s_Jb_J$. Then the directional derivative of $B$ is:
\begin{align}
dB_s(\delta s) &= - A(s) ^{-1} b(\delta s) - A(s) ^{-1} A(\delta s) A(s) ^{-1} b(s)\\ 
&= - A(s) ^{-1} \left( b(\delta s) + A(\delta s)x(s) \right),
\end{align}
where $\delta s$ is an arbitrary vector in $\mathcal{T}_s$. We have used the product rule as well as the derivative formula for matrix inverse: $A^{-1\prime} =- A^{-1}A' A^{-1}$. We claim there are no isolated critical points, and that the critical points form disjoint subsimplices of $\mathcal{T}_s$. If the derivative is degenerate at some $s^*$, there is some non-zero $\delta s ^*$ for which $dB_{s^*}(\delta s^*)=0$. Since $A(s)^{-1}$ is full-rank, this means $b(\delta s^*) + A(\delta s^*)x(s^*) = 0$. Linearity of $b$ and $A$ in $s$ means that $b(s^*+k\delta s^*) + A(s^*+k\delta s^*)x(s^*) = 0$ for all real $k$. But this implies that $x(s^*+k \delta s^*) = x(s^*)$ for all real $k$. It follows that $s^*$ lies in some affine subspace $V_* = s^* + \textrm{ker }dB_{s^*}$ and that every point in $V_*$ is a critical point. There may be more than one critical subspace, but they must be disjoint: a non-zero intersection could be used to generate a higher-dimensional critical subspace that contained the intersecting subspaces. Now if we restrict a critical subspace to $\mathcal{T}_s$, we are left with a subsimplex $\mathcal{T}_*$. We have seen that any critical subsimplex maps to a single point under $B$. We claim that this point lies on the boundary of $B(\mathcal{T}_s)$.

To see why $B(\mathcal{T}_*)\in \partial B(\mathcal{T}_s)$, we show that $\bigcup_{s\in \mathcal{T}_*} \textrm{im }dB_s \ne \mathbb{R}^{n-1}$ which means that there are directions from $B(\mathcal{T}_*)$ that can't be generated by small deviations from $\mathcal{T}_*$. Therefore a neighborhood of $\mathcal{T}_*$ cannot map to a ball in $\mathbb{R}^{n-1}$, which it must if $\mathcal{T}_*$ mapped to the interior of $B(\mathcal{T}_s)$. To determine which direction, let $V_\perp$ be the complementary subspace to $V_*$, so that $\textrm{im }dB_s = \bigcup_{v\in V_\perp} dB_s(v)$. From the formula for $dB_s$ we get that $\textrm{im }dB_s = -A(s)^{-1}(b(V_\perp+A(V_\perp)x_*)$ where $x_*=x(\mathcal{T}_*)$. Since $V_\perp$ has dimension $m<n-1$, $b(V_\perp)+A(V_\perp)x_*$ is an $m$-dimensional linear subspace of $\mathbb{R}^{n-1}$, and there is a  vector $v_{**}$ orthogonal to it. This is the direction we are looking for, because $-A(s)^{-1}$ is a positive-definite matrix, which can never map a vector in a linear subspace to the complement of that subspace (open half-spaces are invariant under positive-definite linear maps). It follows that a sufficiently small neighborhood of $\mathcal{T}_*$ maps to a set that only intersects $\textrm{span }v_{**}$ at $x_*$. Therefore $x_*\not\in \textrm{int }B(\mathcal{T}_s)$.

What we really want to show is that the boundary points of $\mathcal{T}_s$ map to $\partial B(\mathcal{T}_s)$. If $s\in\partial\mathcal{T}_s$ and $s\in V_*$, then we know that $B(s)\in\partial B(\mathcal{T}_s)$, so let us consider a boundary point $s$ that is regular. $s$ cannot map locally to an interior point, so if it maps to an interior point, some other $s'$ must also map there $\emph{i.e.} B(s) = B(s')$. Note however that the structure of $B$ demands that $B(s) = B(s+k(s'-s))$ for any real $k$. This means that $s$ is part of an affine space that maps to $B$. This affine space must be one of the critical subspaces, and so $s$ must map to a boundary point.

Finally, since $\mathcal{A} = \textrm{int } B(\mathcal{T}_s)$, we have $\partial \mathcal{A} = \partial B(\mathcal{T}_s) = B(\partial \mathcal{T}_s)$. To find the boundary of the SLC set, we need only map the boundary points of the simplex $\mathcal{T}_s$.
\end{proof}

The theorem applies only for a control set of $n$ flags, but it can be extended to a larget set:

\begin{corollary}\label{IV_col}
If one uses $n_P>n$ flags as controls, $\mathcal{A} = \bigcup_K B_K(\textrm{int }\mathcal{T}_s)$, where $K$ is a subset of $\{1,\cdots,n_P\}$ with $n$ elements, and $B_K$ is the associated $B$ using $\{b_J, A_J: J\in K\}$. Furthermore, $\partial\mathcal{A}\subseteq \bigcup_K B_K(\partial\mathcal{T}_s)$.
\end{corollary}
\begin{proof}
If $0\in\textrm{int co }\{b_J+A_Jx: J=1,\cdots, n_P\}$, then Carath\'{e}odory's Theorem says that there is a subset $K$ of $n$ indices such that $0\in\textrm{int co }\{b_J+A_Jx: J\in K\}$. We can use the theorem to construct $\mathcal{A}_K$ for each $K$, and Carath\'{e}odory implies that $\mathcal{A}=\bigcup_K \mathcal{A}_K$.  It also follows that $\partial\mathcal{A}\subseteq \bigcup_K B_K(\partial\mathcal{T}_s)$, but equality will typically not hold (the boundary of a union is not necessarily the union of boundaries).
\end{proof}

The preceding theorem can be used to visualize SLC sets for $n=3$ and $4$. We show some examples of this in the following section.

\section{Examples}

The requirement that the $A_J$'s be invertible is not terribly restrictive, as it only requires a certain number of $w^J_{ij}$ be non-zero. For $n=3$, we have:
\begin{align}
\textrm{det} A_J = w^J_{12}w^J_{23}+w^J_{13}w^J_{32}+w^J_{12}w^J_{13}
+w^J_{21}w^J_{13}\nonumber \\+w^J_{23}w^J_{31}+w^J_{21}w^J_{23}+w^J_{31}w^J_{12}
+w^J_{32}w^J_{21}+w^J_{31}w^J_{32}.
\end{align}
Since the $w^J_{ij}$'s are always non-negative, we only need one of nine pairs to be non-zero. 

Theorem \ref{SLCset} states that for any triple of flags, the SLC is the image of $\mathcal{T}_s$ under $B(s)$, which for $n=3$ is a quotient of two homogeneous quadratic functions. Since the boundary of $\mathcal{T}_s$ consists of three line segments, the boundary $\partial\mathcal{A}=B(\partial\mathcal{T}_s)$ consists of three arcs. Now, if we have more than three flags, say $n_c$, the SLC region is the union of the SLC sets for each triple. It follows that there are $n_c \choose  2$ arcs that may contribute to $\partial\mathcal{A}$. If one plots these candidate arcs, we can visualize the SLC region.

For our examples, let $\pi_\iota$ be some complete flag formed out the eigenbasis of the Hermitian operator $A_\iota$. Define $\pi^1,\cdots, \pi^6$ to be the flags obtained by permuting the elements of $\pi_\iota$, so that we have a control-set of six flags. Call this set $\mathbb{F}_\iota$. If $A_\iota$ is simple, it is unique up to re-numbering. This choice of control-set is attractive because all possible tangent vectors at the completely mixed state are contained in the convex hull generated by $\mathbb{F}_\iota$. We have $n_c=6$, and therefore there are fifteen candidate arcs.

Figure \ref{fig1} shows an example for a random Lindblad system. By random, we mean eight Lindblad operators were generated with elements whose real and imaginary parts were uniform on the interval $[0,100]$. The top panel shows the fifteen arcs generated by $\pi_\iota$. The SLC set is the interior of the region formed by these arcs, and this is the dark region shown in the  bottom panel. To get some sense of how ``good" our SLC region is we generated five random unitary matrices, used them to generate five flags as well as their permutations. With these random flags, we used corollary (\ref{IV_col}) to plot a ``better" SLC set. This makes for ${6+6\cdot 5 \choose 2} = 630$ arcs. In the bottom panel of figure \ref{fig1}, we have shown the SLC region for this extended control set as the light region. It is clearly larger, but the original controls cover a good portion. 

Instead of examining random Lindblad systems, we can investigate systems with two specific types of Lindblad operators: jump operators and de-phasing operators. A jump operator relative to a certain orthonormal basis is a Lindblad operator with only one non-zero element, which is off-diagonal. Fix a basis and define, for $j\ne k$,  $L^J_{jk} := \sqrt{\gamma_{jk}}e_{jk}$, where $e_{jk}$ is the matrix with a one at the $(j,k)$ position and zeros elsewhere. Such an operator is called a jump operator as it models a stochastic jump from state $k$ to state $j$. A de-phasing operator meanwhile is a Lindblad operator with only diagonal non-zero elements. It is so-called as any coherent superposition of states will decay to an incoherent mixture so long as the respective diagonal elements are non-zero. In the same basis, define $L^D_l = \sum_{j=1}^n c_{l,j} e_{jj}$, where $l$ indexes the de-phasing operators. Note that with these Lindblad operators, $A_\iota = \sum_{j,k = 1, k\ne j}^3 \gamma_{jk}e_{jj}$. Hence the flag $\pi_\iota$ is in fact generated by the projectors $e_{jj}$.

\begin{figure}
\includegraphics[width=\columnwidth, width=\columnwidth, trim = 0cm -3cm 0cm 0cm]{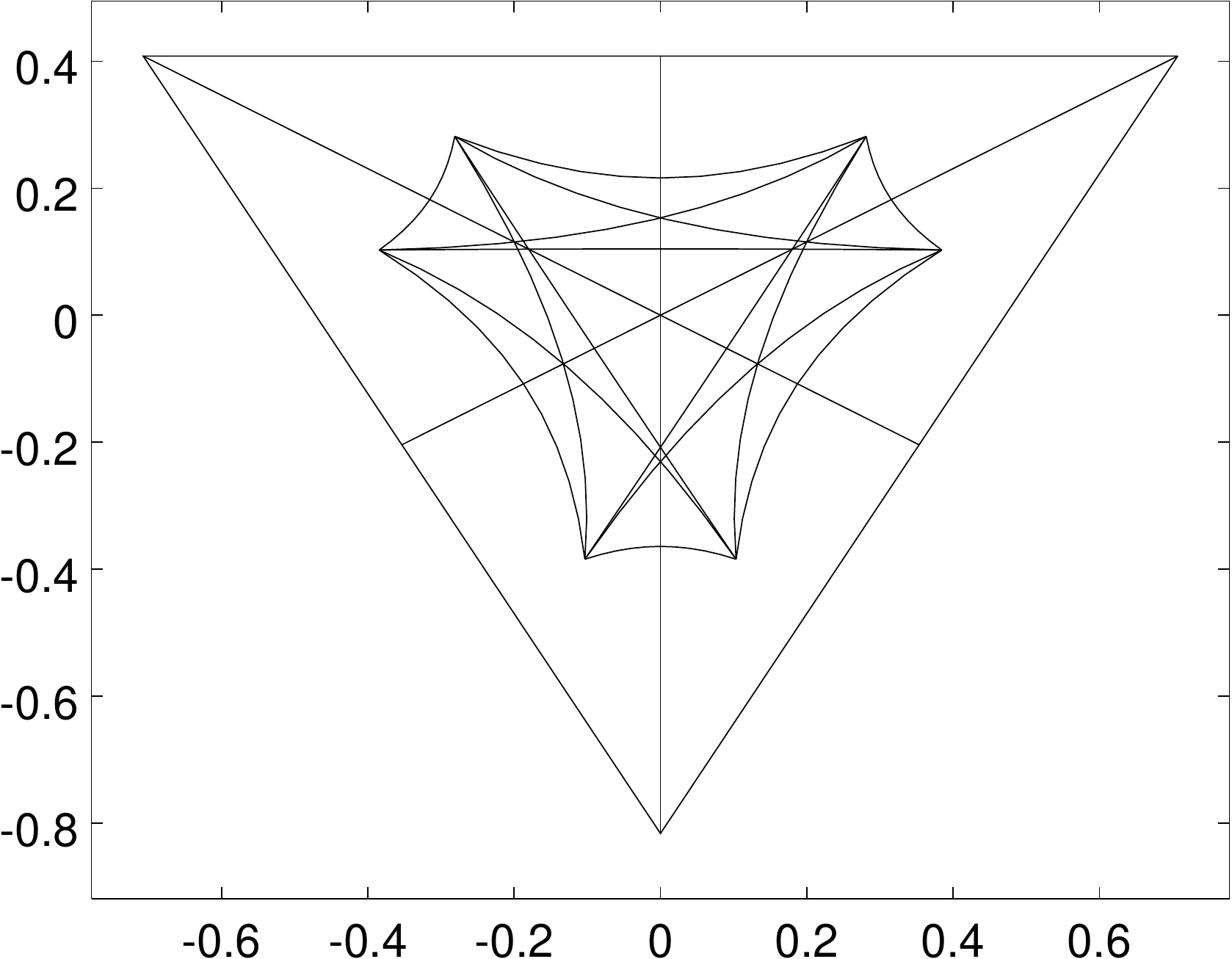}
\includegraphics[width=\columnwidth]{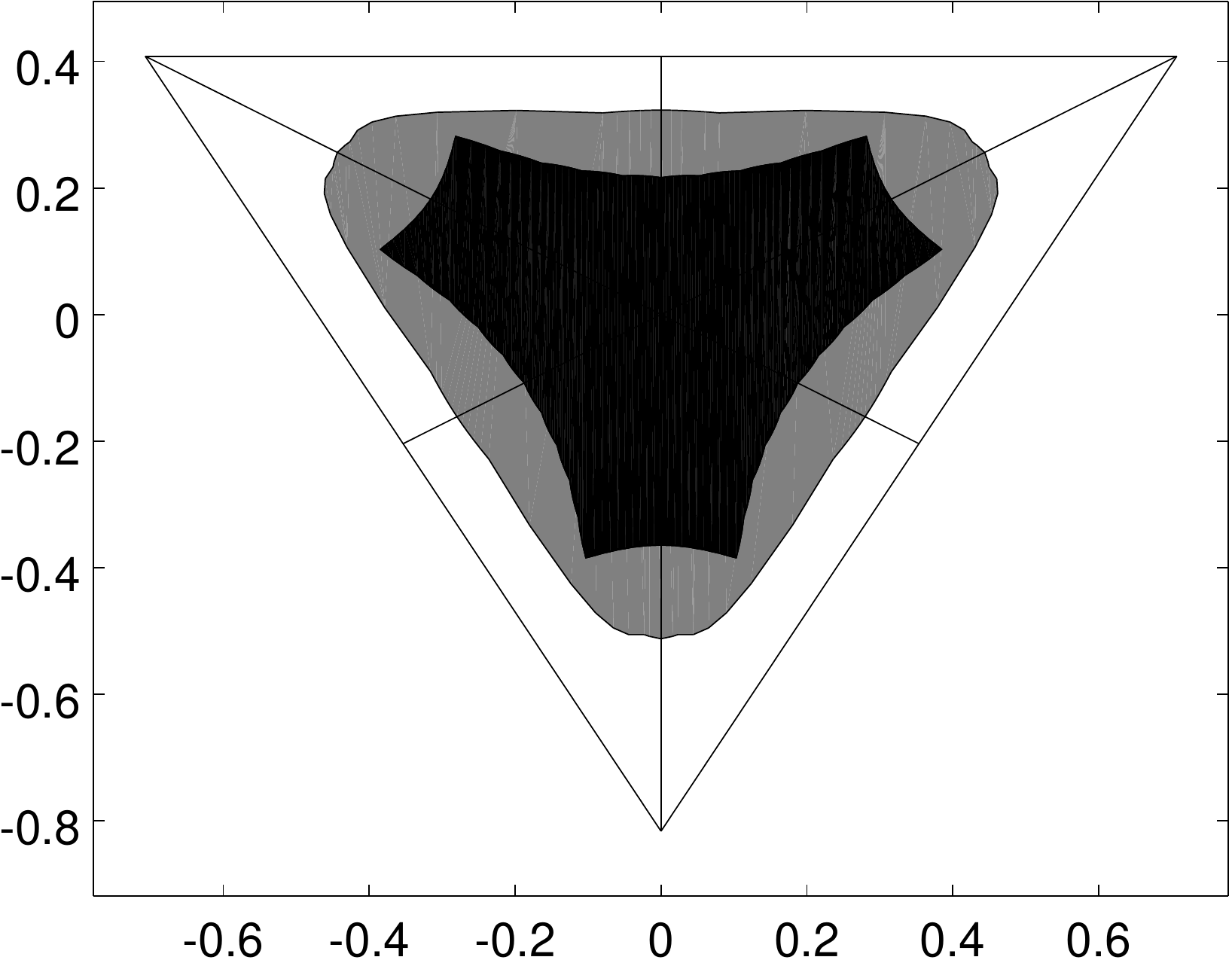}
\caption{(Top) Candidate arcs for $\partial \mathcal{A}$ for a random Lindblad system (Bottom) $\mathcal{A}$ for the same system when (dark) only $\pi_\iota$ are used and (light) five random flags extend the control set. Axes co-ordinates refer to the components of $x$. }
\label{fig1}
\end{figure}

Figure \ref{fig2} shows $\mathcal{A}$ for a system with six jump operators (the coefficients are $\sqrt{\gamma_{12}} = 81$, $\sqrt{\gamma_{77}} = 81$, $\sqrt{\gamma_{23}} = 73$, $\sqrt{\gamma_{32}} = 36$, $\sqrt{\gamma_{31}} = 70$ and $\sqrt{\gamma_{13}} = 48$). The SLC region obtained using $\pi_\iota$, in dark, covers almost the entire SLC region with an extended control set (similar to the preceding example, where there 630 controls in total). This is not an accident. When restricted to jump and de-phasing operators in some basis, it is difficult to find flags other than $\pi_\iota$ and its permutations that enlarge $\mathcal{A}$. The reason for this is that these flags are critical points of the map $w(\pi)$, and in fact the derivative of this map vanishes when $\pi\in\mathbb{F}_\iota$.

To see why, consider that the derivative (\ref{w-deriv}) vanishes if, for each Lindblad operator $L_m$ and component $dw_{jk}$, either $\pi_j L_m \pi_k = 0$, or $\pi_j [L_m, h]\pi_k$ $\forall h\in T_\pi\mathbb{F}$. For the de-phasing operators, the first condition is automatically satisfied, since they are diagonal with respect to the flag $\pi_\iota$. For a jump operator $L^J_{j'k'}$, we have $\pi_j L^J_{j'k'} \pi_k = \delta_{jj'}\delta_{kk'}\sqrt{\gamma_{j'k'}}$, so the first condition is satisfied for all components except for $j=j'$, $k=k'$. And for this component, we claim the second condition is satisfied. 

To see why this claim is true, note that $T_\pi\mathbb{F}$ is the subspace of $\mathfrak{su}(n)$ consisting of all off-diagonal matrices (since any projector set is stationary when acted upon by diagonal matrices). For this reason, $\pi_k h \pi_k \ne 0$, which means $L^J_{jk}h\pi_k =0$. Similarly, $\pi_j h L_{jk}^J = 0$, and therefore $\pi_j [L_{jk}^J, h]\pi_k = 0$. So we can say that $dw(\pi_\iota) = 0$. It follows that the map $\pi\rightarrow \Omega(w(\pi))\Lambda$  has a critical point when $\pi=\pi_\iota$, since $\Omega()$ is linear in $w$.

\begin{figure}
\includegraphics[width=\columnwidth, trim = 0cm -3cm 0cm 0cm]{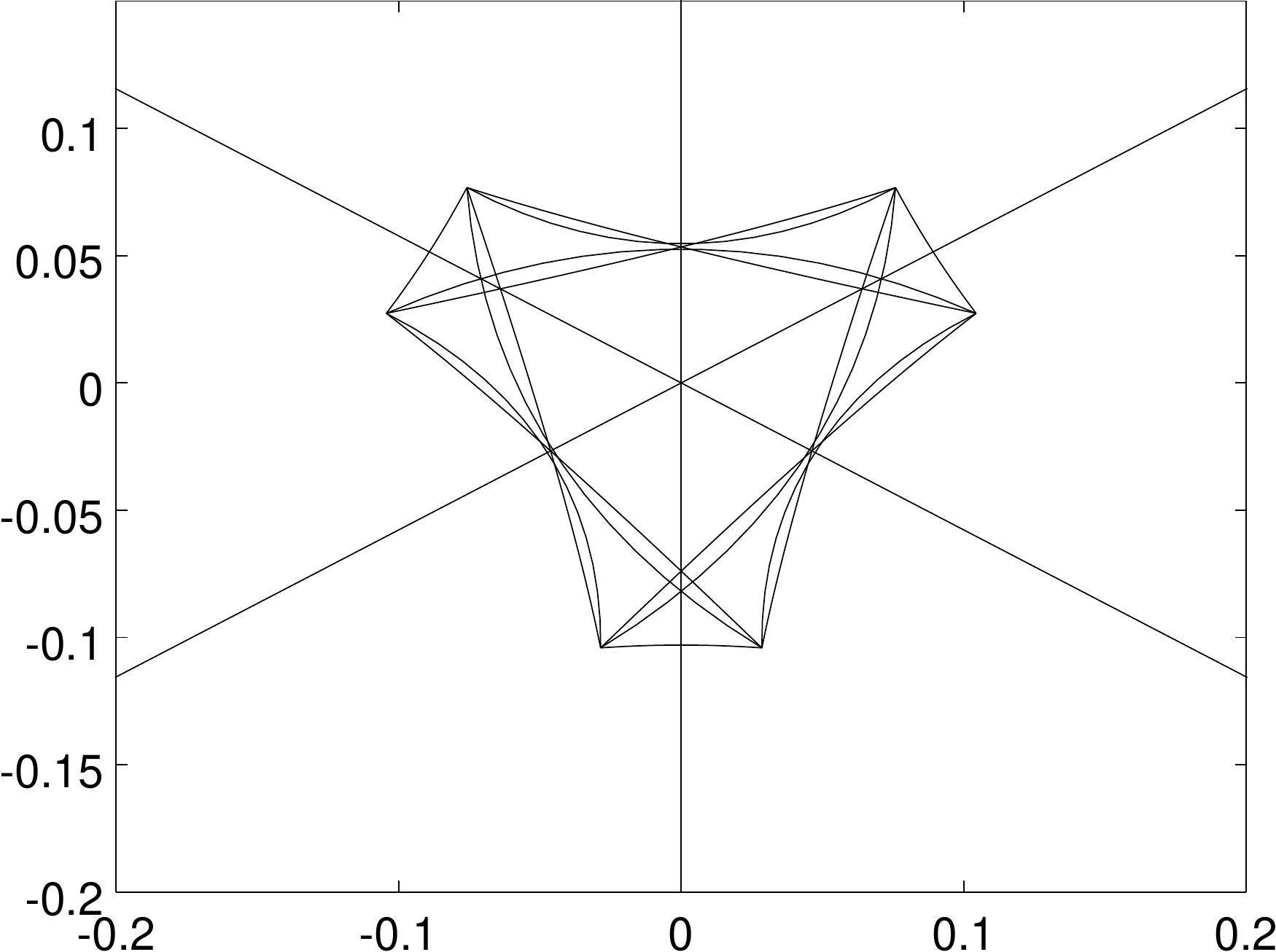}
\includegraphics[width=\columnwidth]{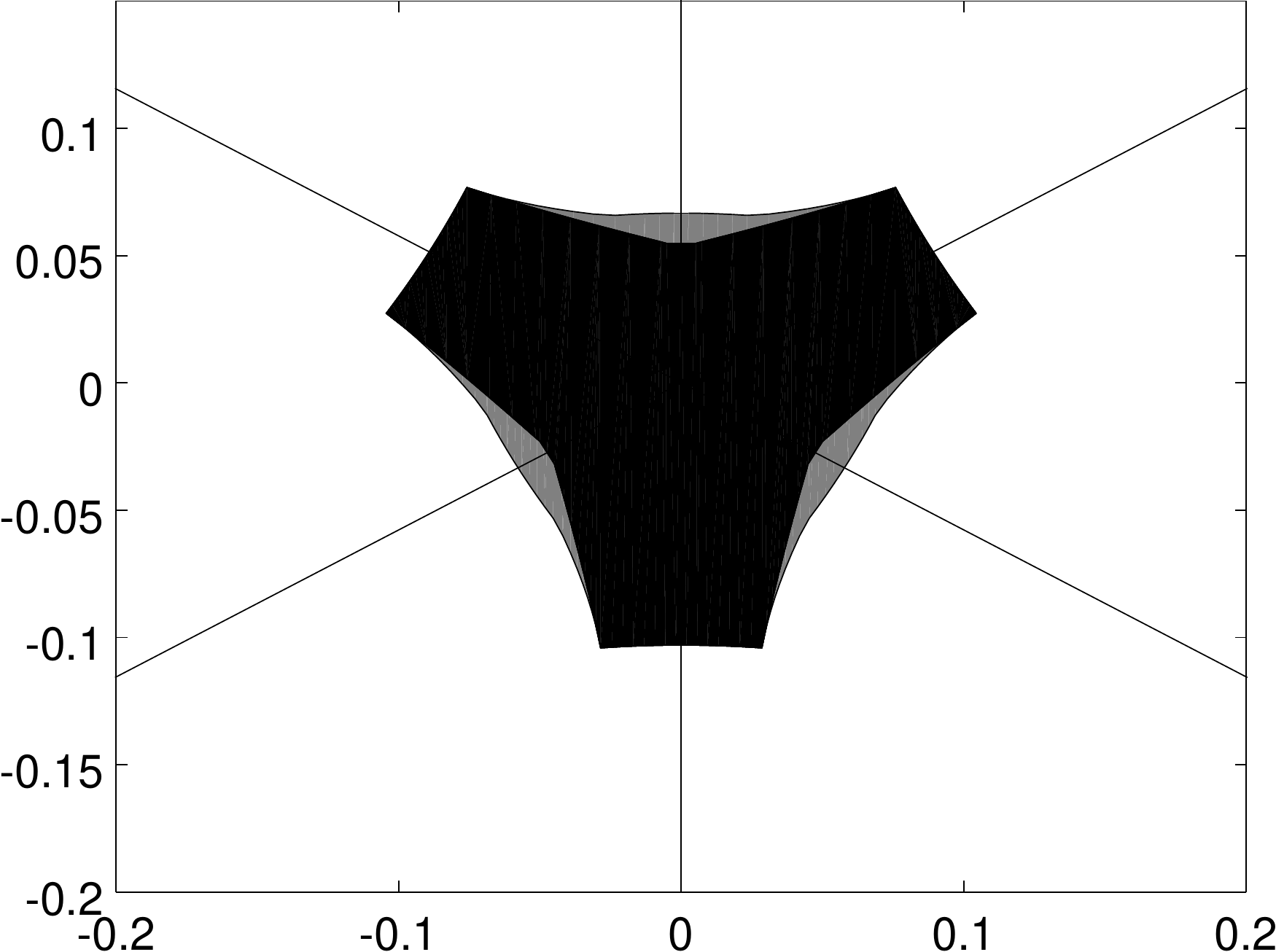}
\caption{(Top) Candidate arcs for $\partial \mathcal{A}$ for a Lindblad system with only jump and de-phasing operators (Bottom) $\mathcal{A}$ for the same system when (dark) only $\pi_\iota$ are used and (light) five random flags extend the control set.}
\label{fig2}
\end{figure}

The significance of $\pi_\iota$ being a critical point is that proposition \ref{SLC} implies that SLC fails when $0$ moves from an interior point of $\mathcal{V}_u(\Lambda)$ to to a boundary point. But a boundary point of $\mathcal{V}_u(\Lambda)$ must be a critical value of $\pi\rightarrow \Omega(w(\pi))\Lambda$, or alternatively a critical value of the map $\pi\rightarrow \Pi\Omega(w(\pi))\Lambda$. Setting $b(\sigma .\pi_\iota) - A(\sigma .\pi_\iota)x=0$ yields the six terminal points of the fifteen arcs from which $\partial \mathcal{A}$ is obtained. Note that in principle, the non-terminal points of the arcs are \emph{not} critical points, but in practice, there is not much room between the arcs and any points that fall outside.

We can also visualize $\mathcal{A}$ for $n=4$. Figures \ref{fig3} and \ref{fig4} show $\partial A$ for two randomly generated systems consisting of only jump operators. Figure \ref{fig3} shows a system with four Lindblad operators: $\sqrt{5}e_{12}$, $\sqrt{3}e_{21}$, $\sqrt{4}e_{23}$ and $\sqrt{3}e_{34}$. For a four-dimensional system, $\mathcal{T}$ consists of twenty-four sub-simplices corresponding to the different eigenvalue orderings. Straight line-segments in the figures are used to indicate the boundaries between the sub-simplices. In figure \ref{fig3}, we see that $\partial\mathcal{A}$ shares a portion of $\partial\mathcal{T}$, but does not include the vertices. The vertices correspond to the orbit of pure states, so it is not possible to purify this system with the flag $\pi_\iota$. However, the edges correspond to states where the two lower eigenvalues are zero, so it is possible to obtain states that are a mixture of only two pure states.

Figure \ref{fig4} has eight Lindblad operators: $\sqrt{4}e_{12}$, $\sqrt{8}e_{13}$, $\sqrt{6}e_{14}$, $\sqrt{13}e_{23}$, $\sqrt{8}e_{32}$, $\sqrt{17}e_{34}$, $\sqrt{4}e_{42}$ and $\sqrt{5}e_{43}$. These have been chosen so that $\partial\mathcal{A}$ includes the orbit of pure states. Interestingly the vertices are the only points on $\partial\mathcal{A}$ that are contained in $\partial\mathcal{T}$. So while it is possible to purify this system with the flag $\pi_\iota$, it is not possible to obtain arbitrary mixtures of two pure states, or even other mixtures of three pure states. .

\begin{figure}
\includegraphics[width=\columnwidth, trim = 0cm -3cm 0cm 0cm]{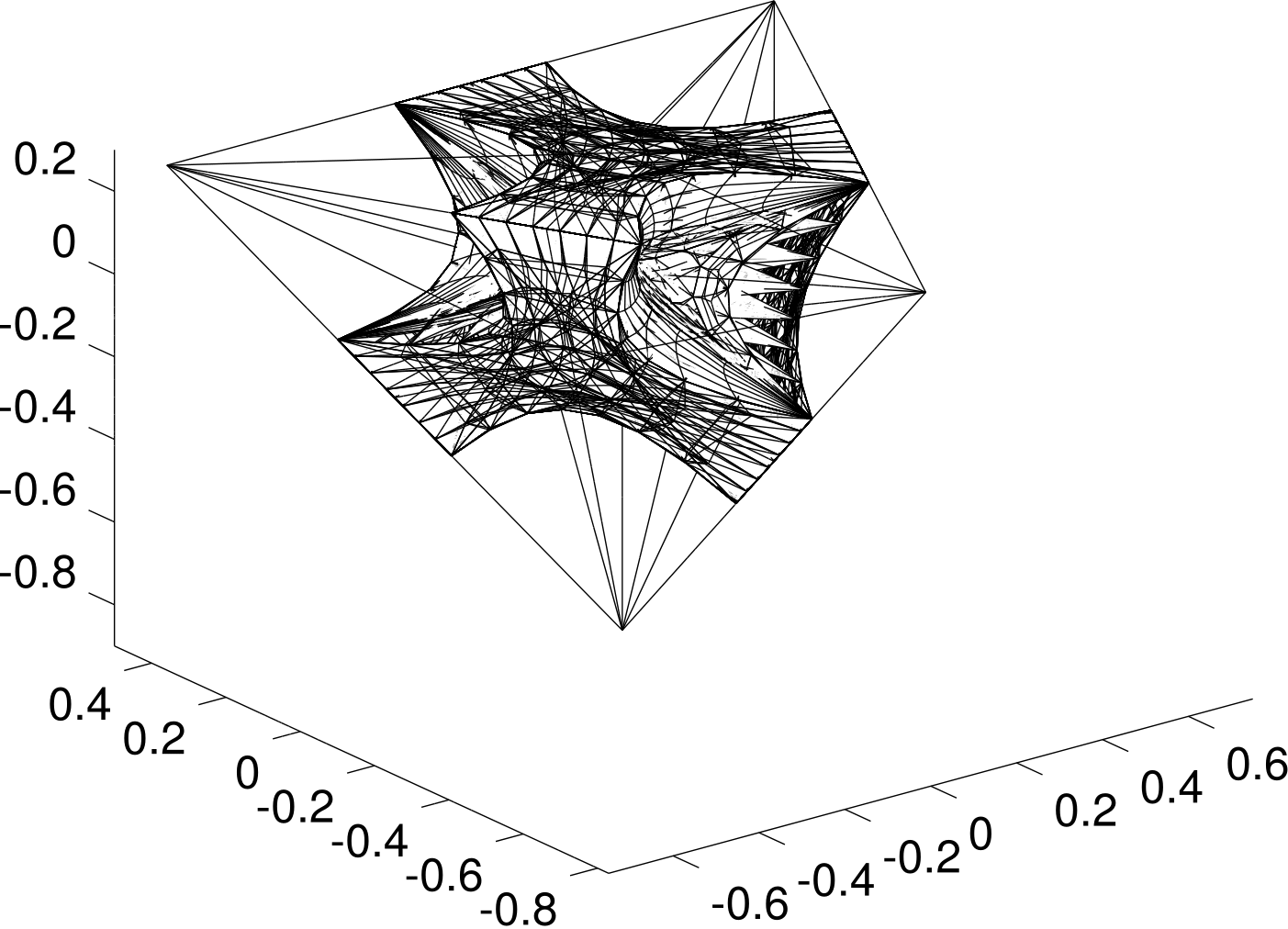}
\caption{$\partial A$ for an $n=4$ Lindblad system that cannot be purified, but for which mixtures of two pure-states are reachable. Axes co-ordinates refer to components of $x$. }
\label{fig3}
\end{figure}

\begin{figure}
\includegraphics[width=\columnwidth, trim = 0cm -3cm 0cm 0cm]{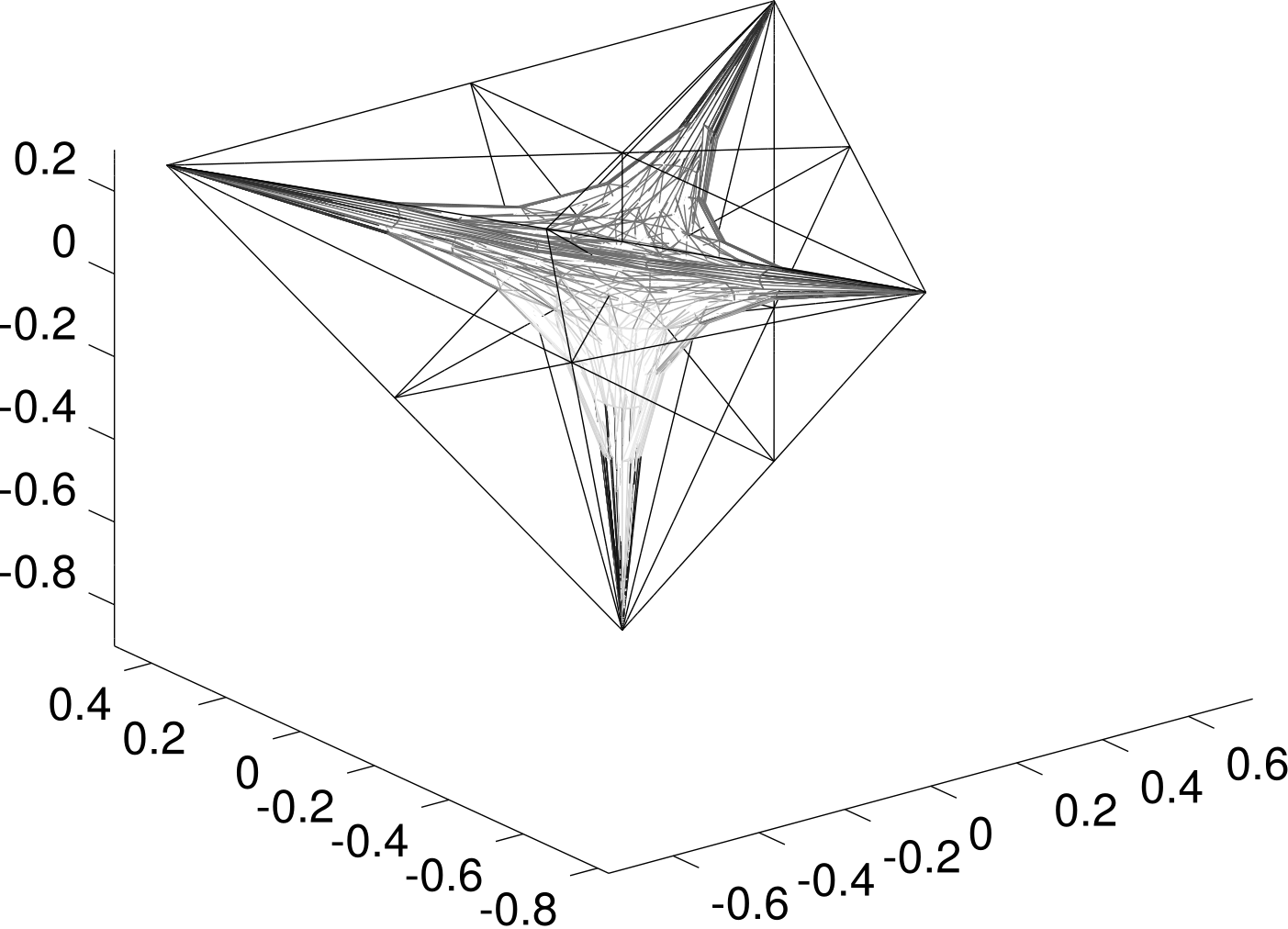}
\caption{$\partial A$ for an $n=4$ Lindblad system for which the only globally reachable mixtures of three pure states are the pure states. }
\label{fig4}
\end{figure}

\section{Conclusions and Future Work}

This paper has demonstrated a procedure by which the dynamics of a quantum Lindblad system can be decomposed into its inter- and intra-orbit dynamics. The purpose of this is to investigate how the system moves between orbits depending on how the system moves along the orbit. Since we can construct arbitrary paths along the orbit given sufficiently fast Hamiltonian control, we would like to know which orbits are reachable, and how to construct the necessary Hamiltonians. We have shown that the orbits can be represented by a state vector $\Lambda$ (technically an equivalence class of such vectors), and the position within the orbit can be represented by a control-flag $\pi$, which is an $n$-tuple of orthonormal projectors. Given this decomposition, we have written down a dynamical equation (\ref{LambODE}) and a control system (\ref{LCS}).We have shown how to reconstruct a Hamiltonian from a desired trajectory along the orbit manifold. Because the orbits are lower-dimensional manifolds at eigenvalue crossings, planning trajectories through crossings require projectors obeying a technical condition.

If one is only studying local controllability, the technicalities concerning eigenvalue crossings can be safely ignored. The challenge in studying local controllability is the fact the control set is not a linear space, but a compact manifold. We have shown that if one limits the control set to a finite subset, the region of strong local controllability can be calculated analytically. We have shown several examples for $n=3$ and $n=4$. While a dramatically smaller control set may appear to be an unnecessary limitation, we have shown for the case where all Lindblad operators are jump and de-phasing operators in a certain basis, almost the entire SLC set can be recovered from a set of $n!$ carefully chosen controls. 

The obvious limitation of this approach is that the control set is highly non-linear and thus it is difficult to attain analytic results. Its compactness however is an attractive feature, and so numerical work may pay dividends. A further drawback to using the analytic result for finite control sets is that the number of hypersurfaces that are candidates for $\partial\mathcal{A}$ grow extremely quickly: there are $n!$ possible $\sigma.\pi_\iota$ and thus the number of hypersurfaces is ${n! \choose n-1} \sim n!^n$. It is only practical for low-dimensional systems, and even for $n=4$, we must construct ${24\choose 3} = 2024$ surfaces (although symmetry makes many of these redundant). Nevertheless, if the Lindblad structure is simple (\emph{i.e.} only one Lindblad operator, or several jump operators), these complications may be mollified. Future work on an numerical extension of this approach is forthcoming.


\begin{acknowledgments}
P.R. has been supported by the National Science Foundation and the DFG grant HE 1858/13-1 from the German Research Foundation (DFG). A.M.B. is supported by the National Science Foundation and the Simons Foundation. C.R. is supported by the Natural Science and Engineering Research Council of Canada.
\end{acknowledgments}

\bibliographystyle{unsrt}
\bibliography{QCbiblio}

\end{document}